\documentclass[12pt,draftclsnofoot,onecolumn]{IEEEtran}
\usepackage[colorlinks,bookmarksopen,bookmarksnumbered,citecolor=red,urlcolor=red,]{hyperref}
\usepackage{CJK}
\usepackage{indentfirst}
\usepackage{algorithm,algorithmic,amsbsy,amsmath,amssymb,epsfig,bbm,mathrsfs, bbm} 
\usepackage{multirow}
\usepackage{mathrsfs}
\usepackage{epstopdf}
\usepackage[noadjust]{cite}

\usepackage{verbatim}
\usepackage{amsfonts}
\usepackage{amsthm}
\usepackage[subfigure]{graphfig}

\newtheorem{theorem}{Theorem}
\newtheorem{prop}{Proposition}
\newtheorem{corollary}{Corollary}

\begin{document}
 
\title{Performance Analysis of Ambient RF Energy Harvesting with Repulsive Point Process Modeling} 

\author{Ian Flint, Xiao Lu, Nicolas Privault, Dusit Niyato, and Ping Wang	\\
\thanks{Ian Flint and Nicolas Privault are with School of Physical $\&$ Mathematical Sciences, Nanyang Technological University, Singapore. (e-mail:iflint@ntu.edu.sg, nprivault@ntu.edu.sg) 
    
Xiao Lu, Dusit Niyato, and Ping Wang are with School of Computer Engineering, Nanyang Technological University, Singapore. (e-mail:luxiao@ntu.edu.sg, dniyato@ntu.edu.sg, wangping@ntu.edu.sg)}
}

\markboth{}{Shell \MakeLowercase{\textit{et al.}}: Bare Demo of
IEEEtran.cls for Journals}

\maketitle

\vspace{-20mm}

\begin{abstract}
Ambient RF (Radio Frequency) energy harvesting technique has recently been proposed as a potential solution to provide proactive energy replenishment for wireless devices. This paper aims to analyze the performance of a battery-free wireless sensor powered by ambient RF energy harvesting using a stochastic geometry approach. Specifically, we consider the point-to-point uplink transmission of a wireless sensor in a stochastic geometry network, where ambient RF sources, such as mobile transmit devices, access points and base stations, are distributed as a Ginibre $\alpha$-determinantal point process (DPP). The DPP is able to capture repulsion among points, and hence, it is more general than the Poisson point process (PPP). We analyze two common receiver architectures: separated receiver and time-switching architectures. For each architecture, we consider the scenarios with and without co-channel interference for information transmission. We derive the expectation of the RF energy harvesting rate in closed form and also compute its variance. Moreover, we perform a worst-case study which derives the upper bound of both power and transmission outage probabilities. Additionally, we provide guidelines on the setting of optimal time-switching coefficient in the case of the time-switching architecture. Numerical results verify the correctness of the analysis and show  various tradeoffs between parameter setting. Lastly, we prove that the sensor is more efficient when the distribution of the ambient sources exhibits stronger repulsion. 

\end{abstract}

\emph{Index terms- Ambient RF energy harvesting, sensor networks, determinantal point process, Poisson point process, Ginibre model}.

\section{Introduction}

Ambient RF energy harvesting techniques offer the capability of converting the received RF signals from environment into electricity \cite{X.2014Lu,XLuSurvey}.  Therefore, it has recently emerged as an alternative method to operate low-power devices \cite{Popovic2013,X.Lu2014}, such as wireless sensors \cite{NParks}. Ambient RF energy harvesting aims to capture and recycle the environmental energy such as broadcast TV, radio and cellular signals~\cite{X.Lu2015}, which are essentially free and universally present, making this technique even more appealing. An experiment with ambient RF energy harvesting in~\cite{A2009Sample} shows that 60$\mu$W is harvested from TV towers that are 4.1km away. It is also reported in~\cite{M2008Tentzeris} that 109$\mu$W RF power can be harvested from daily routine in Tokyo. In \cite{D2010Bouchouicha}, the authors measure the ambient RF power density from 680MHz to 3.5GHz and show that the average power density from 1GHz to 3.5GHz is of the order of 63$\mu$ W/m$^{2}$. Detected 6.3km away from Tokyo Tower, the RF-to-DC conversion efficiency is demonstrated to be about 16$\%$, 30$\%$ and 41$\%$ when the input power is -$15dBm$, -10$dBm$ and -5$dBm$, respectively~\cite{R2003Shigeta}. 

In this context, wireless devices powered by ambient RF energy are enabled for battery-free implementation, and a perpetual lifetime.
For example, reference~\cite{V2013Liu} demonstrates that an information rate of 1kbps can be achieved between two prototype devices powered by ambient RF signals, at distance of up to 2.5 feet and 1.5 feet for outdoors and indoors, respectively. Existing literature has also presented many implementations of battery-free devices powered by ambient energy from WiFi \cite{U2012Olgun}, GSM \cite{M2013Pinuela} and DTV bands  \cite{P2013Nintanavongsa} as well as ambient mobile electronic devices \cite{G2011Karthik}.

\subsection{Related Work}
\label{sec:related}

Geometry approaches have been applied to analyze RF energy harvesting performance in cellular network~\cite{K2014Huang}, cognitive radio network~\cite{S2013Lee}, and relay network~\cite{I2014Krikidis,Z.2014Ding,V.2014Mekikis}. The authors in \cite{K2014Huang} investigate tradeoffs among transmit power and density of mobiles and wireless charging stations which are both distributed as a homogeneous Poisson Point Process (PPP). Energy harvesting relay network has been mostly analyzed. In~\cite{S2013Lee}, the authors study a cognitive radio network where primary and secondary networks are distributed as independent homogeneous PPPs. The secondary network is powered by the energy opportunistically harvested from nearby transmitters in the primary network. Under the outage probability requirements for both coexisting networks, the maximum throughput of the secondary network is analyzed. The study in \cite{I2014Krikidis} analyzes the impact of cooperative density and relay selection in a large-scale network with transmitter-receiver pairs distributed as a PPP. 
Reference~\cite{Z.2014Ding} investigates a decode-and-forward relay network with multiple source-destination pairs. Under the assumption that the relay nodes are distributed as a PPP, the network outage probability has been characterized. 
The studies in \cite{V.2014Mekikis} investigates network performance of a two-way network-coded cooperative network, where the source, destination and RF-powered relay nodes are modeled as three independent PPPs.


Other than RF energy harvesting, stochastic geometry approaches have also been applied to address other types of energy harvesting systems. 
Reference~\cite{K.2014Huang} investigates the network coverage of a hexagonal cellular network, where the base stations are powered by renewable energy, and the mobiles are distributed as a PPP. The authors in~\cite{Y.2014Song} explore the network coverage in a relay-assisted cellular network modeled as a PPP. Each relay node adopt an energy harvesting module, the energy arrival process of which is assumed to be an independent and identical poisson process.                       
In~\cite{S.2014Dhillon}, the authors provides a fundamental characterization of the regimes under which a multiple-tier heterogeneous network with genetic energy harvesting modules  fundamentally achieves the same performance as the ones with reliable energy sources.  Different from above studies, our previous in~\cite{X.LuWCNC} adopts  a determinantal point process model to analyze the downlink transmission performance from an access point to a sensor powered by ambient RF energy.





\subsection{Motivations and Contributions}

As discussed in Section~\ref{sec:related}, the prior literature mainly focuses on the performance analysis on RF-powered wireless devices using PPPs. We generalize this approach by considering a larger class of point processes, specifically the Ginibre $\alpha$-determinantal point process. 
Table~\ref{tab:comparisonpp} compares the PPP and the Ginibre $\alpha$-DPP with regards to a few key points summarized in the table. Therein, the column ``simulation'' refers to the ease (in terms of time, computational complexity and implementation difficulty) of the simulation of the point process. By mathematical tractability, we mean the possibility of obtaining closed mathematical formulas for the moments of the point process. By modulability, we mean the choice of available parameters. 
Lastly, by data fitting we refer to the range of phenomena modeled by the point process. 

\begin{table}
\centering
\caption{\footnotesize Comparison of common stochastic models.} \label{tab:comparisonpp}
\begin{tabular}{|l|l|l|l|l|} 
\hline
Model & Simulation & Mathematical tractability & Modulability & Data fitting\\ \hline
\hline
PPP & Very easy & Closed forms & Intensity & Not very fitting \\
\hline                                   
Ginibre $\alpha$-DPP & Easy & Closed forms & Intensity, $\alpha$ & Different degrees of repulsion\\
\hline                   
\end{tabular}
\end{table}

The Ginibre $\alpha$-DPP offers many advantages in terms of modeling capability and ease of simulation \cite{DecreusefondFlintVergne} (here, $-1\le\alpha<0$ is a parameter, and the PPP is a special case obtained in the limit $\alpha \to 0$). One advantage of the Ginibre $\alpha$-DPP over the PPP is that the Ginibre $\alpha$-DPP can be used to model random phenomena where repulsion is observed. Mobile systems may exhibit some clustering and repulsion behaviors, such as in mobile sensor networks \cite{A.A2007Abbasi}, mobile cellular networks \cite{S2013Cho} and mobile social networks \cite{N2014Vastardis}. Therefore, the Ginibre $\alpha$-DPP is a suitable tool for the analysis of the impact of distribution patterns on network performance. Previous works in \cite{N.2012Miyoshi} and \cite{N.Deng2014} have adopted Ginibre point processes to model the locations of base stations in wireless networks.

Since the direct simulation of DPP models is computationally slow, in this paper we focus on a worst-case scenario which is simpler to analyze. 
Namely, we perform a worst-case analysis of the point-to-point uplink transmission between an RF-powered sensor node and a data sink. The sensor node needs to harvest RF energy from ambient RF sources (e.g., cellular mobiles and access points), which are modeled using a Ginibre $\alpha$-DPP which we have briefly described previously. The sensor, assumed to be battery-free, transmits to the data sink using instantaneously harvested RF energy. A power outage happens if the instantaneously harvested energy fails to meet the circuit power consumption of the sensor. Moreover, if the minimum transmission rate requirement cannot be fulfilled, a transmission outage occurs. For the study of the transmission outage probability, we consider the scenarios of out-of-band transmission and in-band transmission.
In the former scenario, the sensor node transmits data on a  frequency band different from that for RF energy harvesting (without co-channel interference). In the latter scenario, the sensor node transmits on the same frequency band of ambient RF energy sources (with co-channel interference). 
We focus on analyzing the performance of two different receiver architectures: a separated architecture \cite{X.2014Lu} and a co-located receiver architecture, called time-switching \cite{ZhangRuiMIMO}, whose details are in Section~\ref{sec:network_model}. By modeling the RF sources with a Ginibre $\alpha$-DPP, we analyze the impact of the distribution of ambient RF sources on the performance of the RF-powered sensor. Our main contributions are summarized below.

\begin{itemize}

\item First, we derive the expectation of aggregated energy harvesting rate by the sensor. We also obtain the expression of the variance of the energy harvesting rate. 
The closed-form expressions are verified by the numerical results. Since obtaining numerical results is very time-consuming, our closed forms are very useful in practice.


\item Next, we investigate the power outage probability, i.e. the probability that the sensor node becomes inactive due to lack of sufficient energy supply. We give an upper-bound of the power outage probability in closed form, and we interpret this upper-bound as a worst-case scenario. We confirm that the theoretical power outage probability and its estimation by simulation are consistent. We also compare the performance of separated and time-switching receiver architectures by simulation.

\item We further study the transmission outage probability, i.e. the probability that the sensor fails to fulfill its transmission rate requirement, because of insufficient transmit power. The scenarios of out-of-band transmission and in-band transmission are both considered. We derive an upper-bound of the transmission outage probability and again interpret is as a worst-case scenario. 
Inspired by the observations in the numerical results, we derive the optimal value of the time-switching coefficient $\tau$. The expression for the optimal choice of $\tau$ is given and verified by numerical computation. Lastly,
we derive a lower bound of the transmission rate.  

\end{itemize}

Note that this paper is an extension of \cite{ConfVersion}, wherein we present partial of the results with the separated architecture.



The remainder of this paper is organized as follows. Section~\ref{sec:systemmodel} introduces the system model, the DPP geometry model of ambient RF sources and performance metrics. Section~\ref{sec:Analysis} estimates the performance metrics of the sensor for both Ginibre $\alpha$-DPP and PPP modeling of ambient RF sources. 
Lastly, our conclusion can be found in Section~\ref{sec:conclusion}.

{\em Notations}: Throughout the paper, we use $\mathbb{E}[X]$ to denote the probabilistic expectation of a random variable $X$, and $\mathbb{P}(A)$ to denote the probability of an event $A$.

\section{System Model}
\label{sec:systemmodel}



\subsection{Network Model}
\label{sec:network_model}

\begin{figure}
\centering
\includegraphics[width=0.4\textwidth]{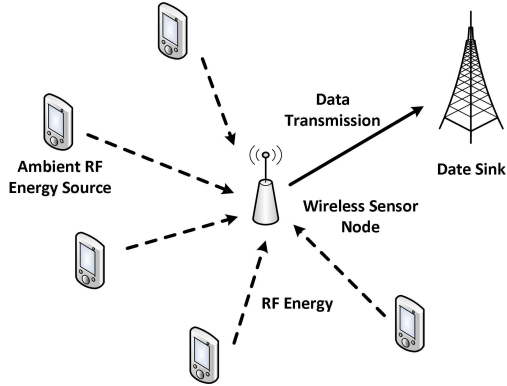}
\caption{A network model of ambient RF energy harvesting.} \label{fig:systemmodel}
\end{figure}


We consider a network comprising a number of ambient RF energy sources, i.e., wireless information transmitters, in which a sensor node is powered solely by the energy harvested from these energy sources. Figure \ref{fig:systemmodel} shows the considered network model, where the sensor node harvests RF energy emitted from the ambient sources and utilizes the harvested energy to perform  uplink data transmission to the data sink.
We model the distribution of ambient RF energy sources as a Ginibre $\alpha$-DPP, which will be specified in detail in Section~\ref{sec:geometricmodeling}. The transmit power of the ambient RF sources are assumed to be identical. Without any loss of generality, the sensor is considered to lie at the origin. Furthermore, we assume that the sensor node is battery-free. In particular, the sensor utilizes the instantaneously harvested RF energy to supply its operations. 
We study two different receiver architectures: separated receiver architecture, and time-switching, which either enables the sensor to perform data transmission and RF energy harvesting simultaneously or separately.

\begin{itemize}

\item{Separated receiver architecture}: As shown in Fig. \ref{SRA}, this architecture equips the energy harvester and the information transmitter with separated antennas so that they can function independently and concurrently. The instantaneously harvested energy is first used to operate the sensor circuit and then the surplus energy is provided for information transmission. This architecture can maximize the utilization of energy harvesting devices, but it is generally larger in size compared to the time-switching architecture.

\item{Time-Switching Architecture}: As shown in Fig. \ref{TSA}, the time-switching architecture,  is equipped with a single antenna. By adopting a switcher, this architecture allows either the energy harvester or the information transmitter attached to the antenna at a time. The time-switching architecture works on a time-slot basis. In each time slot, the energy harvester first uses $\tau$ ($0\leq\tau\leq 1$) portion of a time slot to harvest RF energy. The capacitor reserves the surplus of the harvested energy after being used to power the sensor circuit.  Next, during the rest of $1-\tau$ time, the information transmitter utilizes the surplus energy from the capacitor to transmit information.


\end{itemize}



\begin{figure} 
\centering
\subfigure [Separated Receiver Architecture] {
 \label{SRA}
 \centering
 \includegraphics[width=0.35 \textwidth]{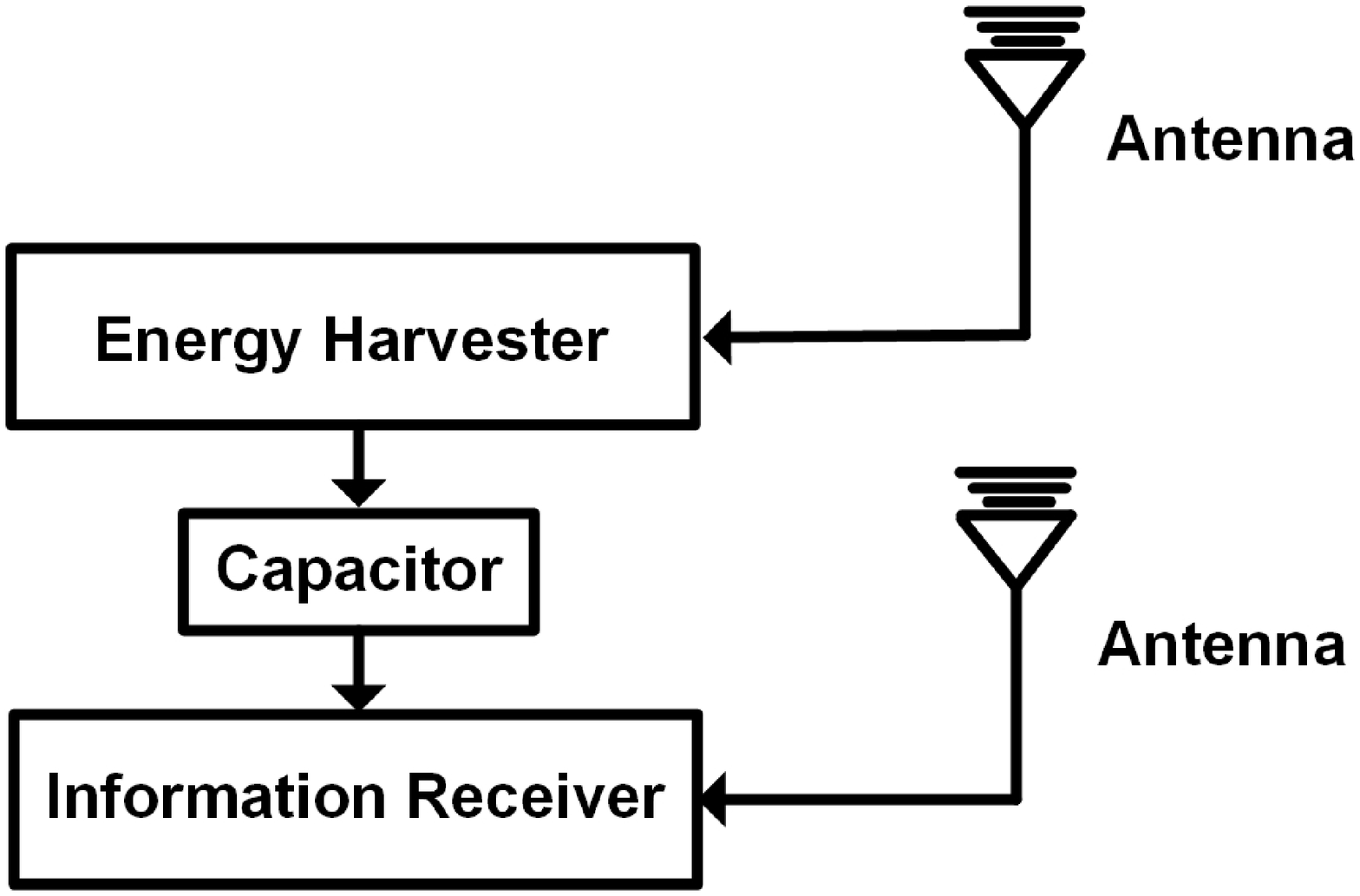}}  
 \centering
 \subfigure [Time Switching Achitecture] {
  \label{TSA}
  \centering
  \includegraphics[width=0.48  \textwidth]{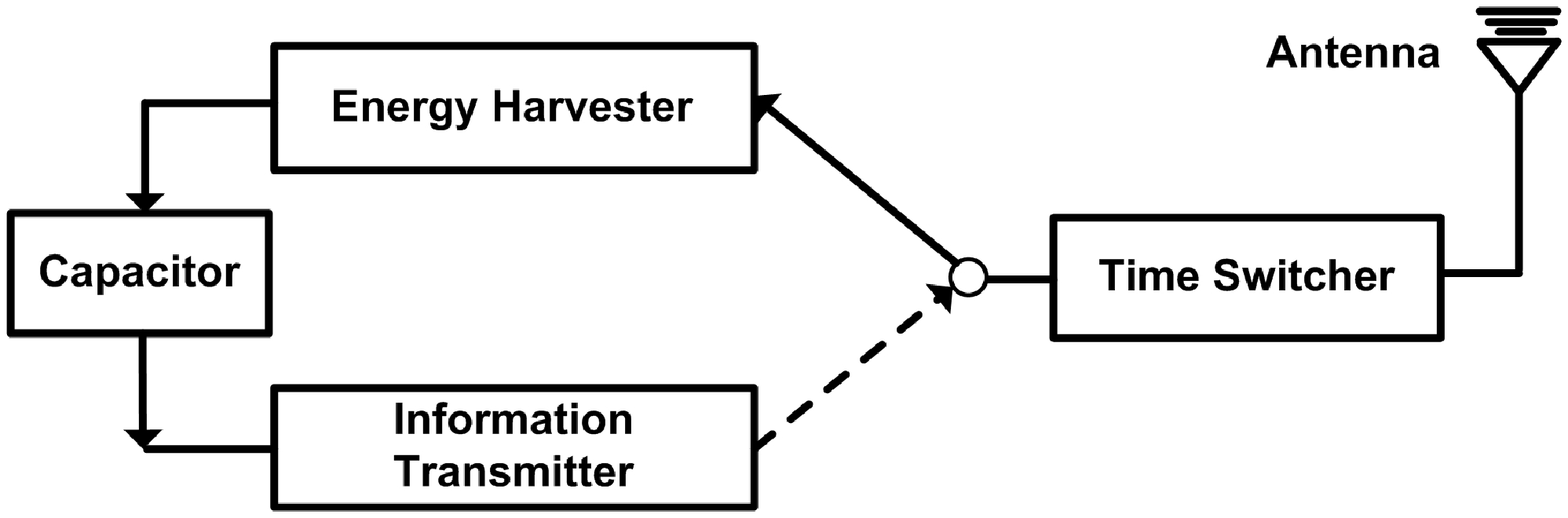}}   
  \centering
\caption{Separated receiver and time switching architectures.} 
\label{receiver_designs}
\end{figure}


The RF energy harvesting rate of the sensor node from the RF energy source $k$ in a free-space channel $P^{k}_{\mathrm{H}}$ can be obtained based on the Friis equation~\cite{Visser2013} \footnote{Other RF signal propagation models can also be used without loss of generality in the analysis of this paper.} as follows: 
\begin{equation}
\label{eq:harvestedRFpower}
	P^{k}_{\mathrm{H}}	= \varrho \beta P_{\mathrm{S}} \frac{G_{\mathrm{S}} G_{\mathrm{H}} \lambda^{2}}{(4\pi d_{k})^{2}},
\end{equation}
where $\beta$ is the RF-to-DC power conversion efficiency of the sensor node, and $\varrho$ is an efficiency factor which depends on the specific architecture.
For a given RF energy source, $P_{\mathrm{S}}$ is its transmit power, $G_{\mathrm{S}}$ is its transmit antenna gain, $\lambda$ is the wavelength at which it emits. As the focus of this paper is to analyze the impact of the locations of ambient RF sources to the performance of the sensor node, we intentionally make other parameters, e.g., $P_S$, $G_S$, and $\lambda$ to be constants for ease of presentation and analysis. Nevertheless, the proposed analytical framework can also be extended to the case when these parameters vary. $d_{k}$ is the distance of an RF energy source $k$ to the receiver antenna of the sensor node. $G_{\mathrm{H}}$ is the receive antenna gain of  the sensor node. Let ${\mathbf{x}}_k \in {\mathbb{R}}^2$ be the coordinates of the RF energy source $k$ (recall that the sensor node lies at the origin). The distance is modeled as $d_{k} = \epsilon + \lVert{\mathbf{x}}_k	\rVert$, where $\epsilon$ is a fixed (small) parameter which ensures that the associated harvested RF power is finite in expectation. Physically, $\epsilon$ is the closest distance that the RF energy sources can be to the sensor node.

The aggregated RF energy harvesting rate by the sensor node from the ambient RF sources can be computed as
\begin{equation}
\label{eq:totalamountofpower}
	P_{\mathrm{H}} = \sum_{ k \in {\mathcal{K}} } P^{k}_{\mathrm{H}}	=	\sum_{ k \in {\mathcal{K}} }	\varrho \beta P_{\mathrm{S}} \frac{G_{\mathrm{S}} G_{\mathrm{H}} \lambda^{2}}{(4\pi 	(\epsilon + \lVert {\mathbf{x}}_k \rVert)	)^{2}} ,
\end{equation}
where ${\mathcal{K}}$ is a random set consisting of all RF energy sources. We assume that $\mathcal{K}$ is a point process~\cite{Kallenberg}.


The sensor consumes a base circuit power, denoted by $P_{\mathrm{C}}$. Note that this circuit power consumption also accounts for the energy loss due to various factors such as capacity leakage. Following practical models~\cite{G2009Miao}, the circuit power consumption of the sensor is assumed to be fixed. We assume that, other than circuit power consumption, there is no power loss during transfer from energy harvester to information transmitter for both architectures.
For the separated receiver architecture, the transmit power is given by $P_{\mathrm{T}} = \left[	P_{\mathrm{H}} - P_{\mathrm{C}} \right]^+$, where $[x]^+ = \max(0, x)$ and $P_{\mathrm{C}}$ is a constant. For the time-switching architecture, all the harvested energy can be used in the data transmission phase to maximize transmission rate. Thus the transmit power is dependent on the transmission time, and is given by $P_{\mathrm{T}} = \left[	P_{\mathrm{H}} - P_{\mathrm{C}} \right]^+/(1-\tau)$.
Then, the general form of maximum transmission rate of the sensor node is given as follows\footnote{Note that state-of-the-art wireless information receivers are not yet able to achieve this rate upper bound due to additional processing noise such as the RF band to baseband conversion noise.}:
\begin{equation}
\label{eq:maxtransmission}
	C	=	\eta\cdot W \cdot \log_2	\left(	1 + h_0	\frac {	\left[	P_{\mathrm{H}} - P_{\mathrm{C}} \right]^+	}	{\eta(\xi\sum_{k\in\mathcal K}P_{\mathrm H}^k+\sigma^2)} 	\right),
\end{equation}
where $W$ is the transmission bandwidth, and $0\le\eta\le 1$ is an efficiency factor depending on the specific architecture. $\sigma^2$ is a nonnegative constant which represents the power of additive white Gaussian noise (AWGN). The term $\xi\sum_{k\in\mathcal K}P_{\mathrm H}^k$ corresponds to the interference, and the specific value of $\xi\in\{0,1\}$ depends on whether we consider an out-of-band or in-band transmission scenario. $h_0$ denotes the channel gain between the transmit antenna of the sensor node and the receive antenna of data sink. 
The separated receiver architecture corresponds to $\varrho=1$ and $\eta=1$.
The time-switching architecture corresponds to the $\varrho=\tau$ and $\eta=1-\tau$, where $\tau$ is the time-switching parameter.
In both of these cases, $\xi=1$ corresponds to an in-band transmission scenario, while $\xi=0$ corresponds to an out-of-band transmission.

\subsection{Geometric DPP Modeling of Ambient RF Energy Sources}
\label{sec:geometricmodeling}

As an extension of the Poisson setting, we model the locations of RF energy sources using a point process $\mathcal{K}$ on an observation window $O\subset\mathbb R^2$ such that $0<|O|<+\infty$ (here $|O|$ denotes the Lebesgue measure of $O$). In other terms, $\mathcal{K}$ is an almost surely finite random collection of points inside $O$. We refer to~\cite{Kallenberg} and~\cite{DaleyVereJones} for the general theory of point processes. The correlation functions $\rho^{(n)}$ of $\mathcal K$ (if they exist), {\em w.r.t.} the Lebesgue measure on $\mathbb R^2$,  verify
\begin{equation}
\label{eq:defcorrelation}
\mathbb E\left[\prod_{i=1}^{n}\mathcal K(B_i)\right]=\int_{B_1\times\cdots\times
B_n}\rho^{(n)}(x_1,\ldots,x_n)\,\mathrm{d}x_1 \cdots	\mathrm{d}x_n,
\end{equation}
for any family of mutually disjoint bounded subsets $B_1,\ldots,B_n$ of $E$, $n\geq 1$. 
Heuristically, $\rho^{(1)}$ is the particle density, and 
$\rho^{(n)}(x_1,\ldots,x_n)\,\mathrm dx_1\ldots\mathrm dx_n$
is the probability of finding a particle in the vicinity of each $x_i$, $i=1,\dots,n$. The correlation functions are thus a generalization of the concept of probability density function to the framework of point processes. The correlation functions play an important role in the definition and interpretation of a general $\alpha$-DPP.

\subsubsection{General $\alpha$-determinantal point process}

We let $\alpha=-1/j$ for an integer $j\in\mathbb{N}^*$, and we define a general $\alpha$-DPP in the following. Let us introduce a map $K:\mathrm L^2(\mathbb R^2)\mapsto\mathrm L^2(\mathbb R^2)$, where $\mathrm L^2(\mathbb R^2)$ is the space of square integrable functions on $\mathbb R^2$. We assume in the following that $\mathcal K$ satisfies Condition~A from \cite{ShiraiTakahashi}.
The map $K$ is called the {\em kernel} of the $\alpha$-DPP. It represents the interaction force between the different points of the point process. 
A locally
finite and simple point process on $\mathbb R^2$ is called an $\alpha$-DPP if its
correlation functions {\em w.r.t.} the Lebesgue measure on $\mathbb R^2$ exist and
satisfy
\begin{equation}
\label{eq:correlationfunctions}
\rho^{(n)}(x_1,\ldots,x_n)=\mathrm{det}_\alpha(K(x_i,x_j))_{1\leq i,j\leq n},
\end{equation}
for any $n\geq 1$ and $x_1,\ldots,x_n\in\mathbb R^2$, and where the $\alpha$-determinant of a matrix $M=(M_{ij})_{1\le i,j\le n}$ is defined as
\begin{eqnarray}
\mathrm{det}_\alpha\,M = \sum_{\sigma\in S_n} \alpha^{n - \nu(\sigma)} \prod_{i=1}^n M_{i \sigma(i)},
\end{eqnarray}
where $S_n$ stands for the $n$-th symmetric group and $\nu(\sigma)$ is the number of cycles in the permutation $\sigma \in S_n$. 

Let us now give some basic properties of the $\alpha$-DPP to emphasize the role played by the kernel $K$. We start by a proposition exhibiting the repulsion properties of the $\alpha$-DPP. Its proof follows from the formula defining the correlation functions \eqref{eq:defcorrelation}.
\begin{prop}[Repulsion of the $\alpha$-DPP]
 The covariance of an 
 $\alpha$-DPP of kernel $K$ is given by 
\begin{equation*}
\mathrm{Cov}(\mathcal{K}(A),\mathcal{K}(B))=\alpha\int_{A\times B}|K(x,y)|^2\,\mathrm{d}x\mathrm{d}y, 
\end{equation*}
 where 
 $\mathcal{K}(A)$ and $\mathcal{K}(B))$ 
 denote the random number of point process 
 points located within the disjoint bounded sets 
 $A,B\subset\mathbb{R}^2$. 
\end{prop}

Since $\alpha < 0$, $\mathcal{K}(A)$ and $\mathcal{K}(B)$ are negatively correlated and the associated $\alpha$-DPP is known to be locally Gibbsian, see, e.g.,~\cite{GeorgiiYoo}, therefore it is a type of repulsive point process. Additionally, the $\alpha$-DPP exhibits more repulsion when $\alpha$ is close to $-1$. As $\alpha\rightarrow 0$, $\mathcal{K}(A)$ and $\mathcal{K}(B)$ tend not to be correlated, and in fact it can be shown that the corresponding point process converges weakly to the PPP, c.f.~\cite{ShiraiTakahashi}.

Next, we recall from \cite{Soshnikov} the following proposition which gives the hole probabilities of the $\alpha$-DPP.
Proposition~\ref{prop:holeproba} allows us to compute the quantities known as hole probabilities. 
\begin{prop}[Hole probability of the $\alpha$-DPP]
\label{prop:holeproba}
 For every bounded set $B\subset\mathbb{R}^2$ we have 
\begin{equation}
\label{aeq} 
 \mathbb P(\mathcal{K}\cap B = \emptyset ) = \mathrm{Det}(\mathrm{Id}+\alpha K_B)^{-1/\alpha},
\end{equation}
 where $K_B$ is the operator restriction of $K$ to the space $\mathrm{L}^2(B)$ of square integrable functions on $B$ with respect to the Lebesgue measure. Here, $\mathrm{Id}$ is the identity operator on $\mathrm{L}^2(B)$ and for any trace class integral operator $K$, $\mathrm{Det}\left(\mathrm{Id}+\alpha K\right)$ is the Fredholm determinant of $\mathrm{Id}+\alpha K$ defined in \cite{Brezis}.

\end{prop}

\subsubsection{The Ginibre point process}

In the rest of the paper, we focus on the Ginibre $\alpha$-DPP, which is a particular $\alpha$-DPP well-suited for applications. The Ginibre process is a type of $\alpha$-DPP that is invariant with respect to rotations. Therefore, it will be fruitful for computational convenience to restrict our attention to the choice of observation window $O=\mathcal B(0,R)$, defined as a disc centered around $0$ and of radius $R>0$.

The Ginibre process is defined by the so-called Ginibre kernel given by
\begin{equation}
\label{eq:ginibre}
K(x,y)=\rho\,e^{\pi\rho x \bar{y}} e^{-\frac{\pi\rho}{2}( |x|^2 + |y|^2)},
\quad 
 x,y \in O=\mathcal{B}(0,R),
\end{equation}
where $\rho>0$ is a fixed parameter called {\em density} of the point process.
This kernel is that of the usual Ginibre process defined, e.g., in~\cite{DecreusefondFlintVergne}, to which we have applied a homothety of parameter $\sqrt{\pi\rho}>0$: $x\mapsto x/(\sqrt{\pi\rho})$. Next we recall a few features of the Ginibre process. 
\begin{itemize} 
\item
The Ginibre kernel $K$ defined in \eqref{eq:ginibre} satisfies Condition~A from \cite{ShiraiTakahashi}, and is thus a type of $\alpha$-DPP.
\item The intensity function of the Ginibre process 
 is 
 given by 
\begin{equation} 
\label{beq} 
\rho^{(1)}(x)=K(x,x)= \rho, 
\end{equation} 
 c.f.~\cite{ShiraiTakahashi}. This means that the average number of points in a bounded set $B\subset \mathcal{B}(0,R)$ is $\rho\,|B|$. Note that the intensity function of a homogeneous PPP is also constant, so $\rho$ is interpreted as the intensity of the corresponding PPP.
\\ 
 
\item 
The Ginibre $\alpha$-DPP is stationary and isotropic, in the sense that its distribution is invariant with respect to translations and rotations, c.f.~\cite{DecreusefondFlintVergne}. Hence, the Ginibre point process models a situation where the RF energy sources are distributed homogeneously in $\mathbb R^2$.
\end{itemize} 

 We write $\mathcal{K}\sim\mathrm{Gin}(\alpha,\rho)$ 
 when $\mathcal{K}$ is an $\alpha$-DPP with Ginibre
 kernel defined in \eqref{eq:ginibre} and density $\rho$. 
 The spectral theorem for Hermitian and compact operators 
 yields the decomposition 
$
K(x,y)=\sum_{n\ge 0} \lambda_n \varphi_n(x)\overline{\varphi_n(y)},
$
where $(\varphi_i)_{i\ge 0}$ is a basis of eigenvectors of $\mathrm{L}^2(O)$, and $(\lambda_i)_{i\ge 0}$ are the corresponding eigenvalues. 
 In, e.g.,~\cite{DecreusefondFlintVergne}, it is shown that
 the eigenvalues of the Ginibre point process on $O=\mathcal{B}(0,R)$ 
 are given by 
\begin{equation}
\label{eq:eigenvalues}
\lambda_n = \frac{\Gamma(n+1, \pi\rho R^2)}{n!},\qquad n\in\mathbb N,
\end{equation}
 where 
\begin{equation}
\label{eq:defgamma}
\Gamma(z,a) \triangleq \int_0^a e^{-t} t^{z-1}\,\mathrm{d}t, 
\qquad 
z \in \mathbb{C}, \quad a \ge 0, 
\end{equation}
 is the lower incomplete Gamma function. 
 On the other hand, the eigenvectors of $K$ are given by
$
\varphi_n(z) \triangleq  \frac{1}{\sqrt{\lambda_n}}\frac{\sqrt{\rho}}{\sqrt{ n!}} e^{-\frac{\pi\rho}{2} | z |^2} (\sqrt{\pi\rho} z)^n, 
$
for $n\in\mathbb N$ and $z\in O$.
We refer to~\cite{DecreusefondFlintVergne} for 
further mathematical details on the Ginibre point process. 

To illustrate how the parameter $\alpha$ affects the distribution of the DPP, in Fig.~\ref{Distribution} we show some snapshots of the scattering of ambient RF energy sources in a disc of radius $R=10$, when the RF source density is $\rho=0.3$. It is seen that strong repulsion exists between the RF sources when $\alpha=-1$. As a result, the RF sources tend to scatter evenly over the area. We can observe that the repulsion decreases very fast with the increase of $\alpha$. When $\alpha=-0.5$, some of the RF sources exhibit attraction by locating close to each other. Some grids of vacant area begin to emerge. The attraction
keeps increasing as $\alpha$ approaches zero. When $\alpha=-0.03$, the RF sources show clustering behavior, which is a feature of the PPP. Consequently, there appear many grids of vacant area. Thus, depending on the distribution of RF sources that is observed, we shall choose a different value of the parameter $\alpha$ that appropriately models the situation at hand. 

\begin{figure*}[t]
\begin{center}
$\begin{array}{ccc} 
\epsfxsize=2.0 in \epsffile{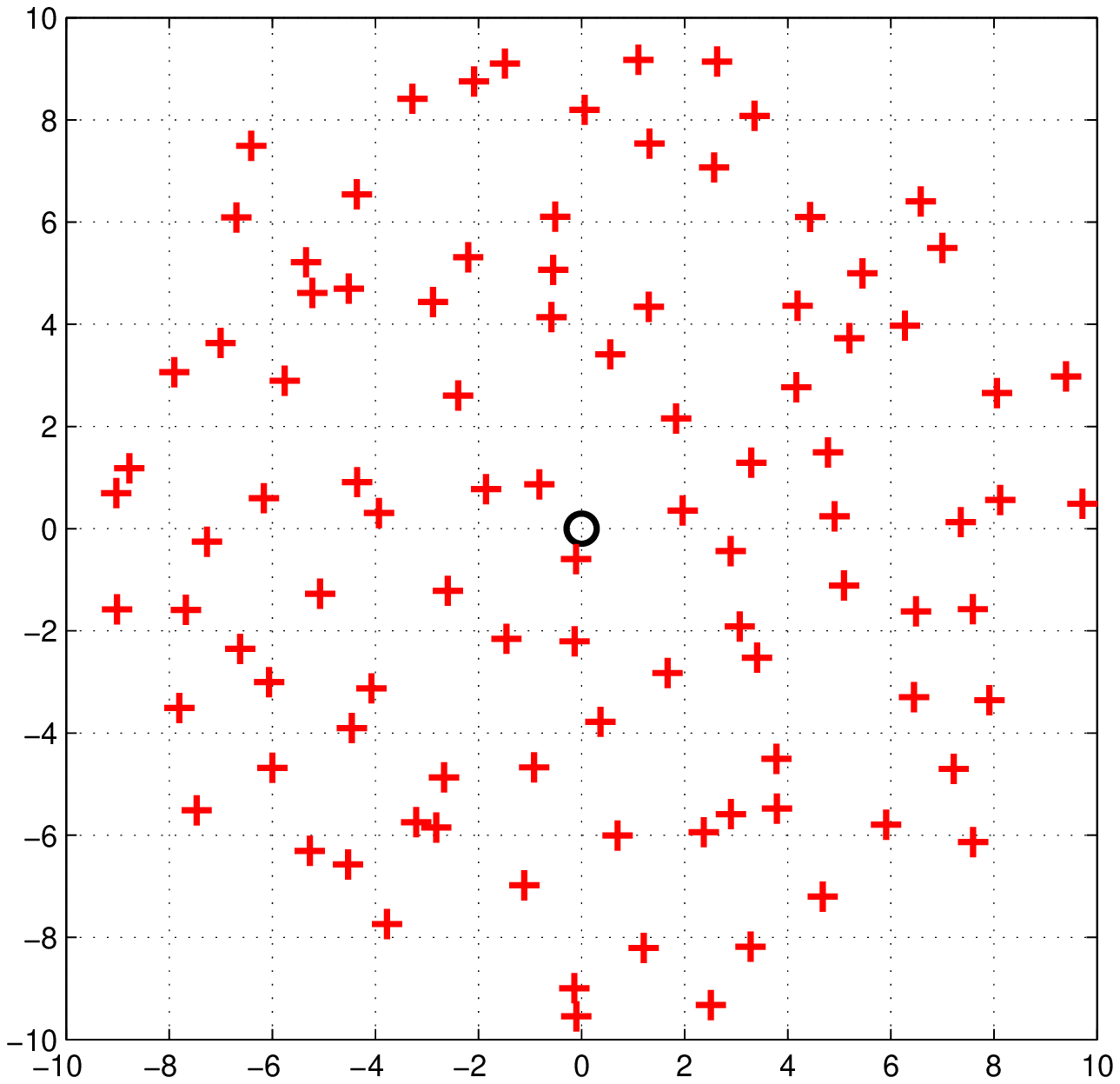}	&
\epsfxsize=2.0 in \epsffile{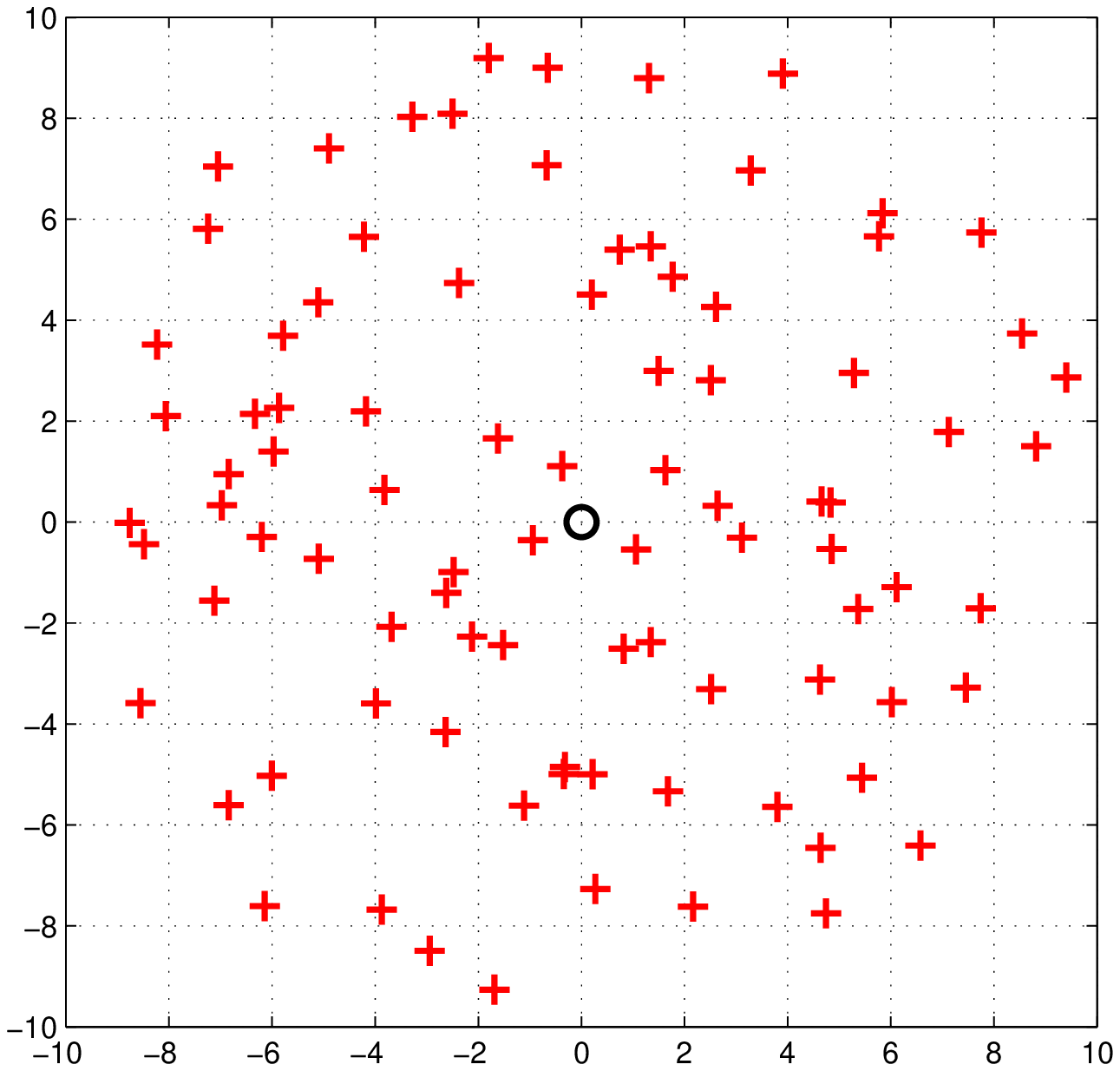}	&
\epsfxsize=2.0 in \epsffile{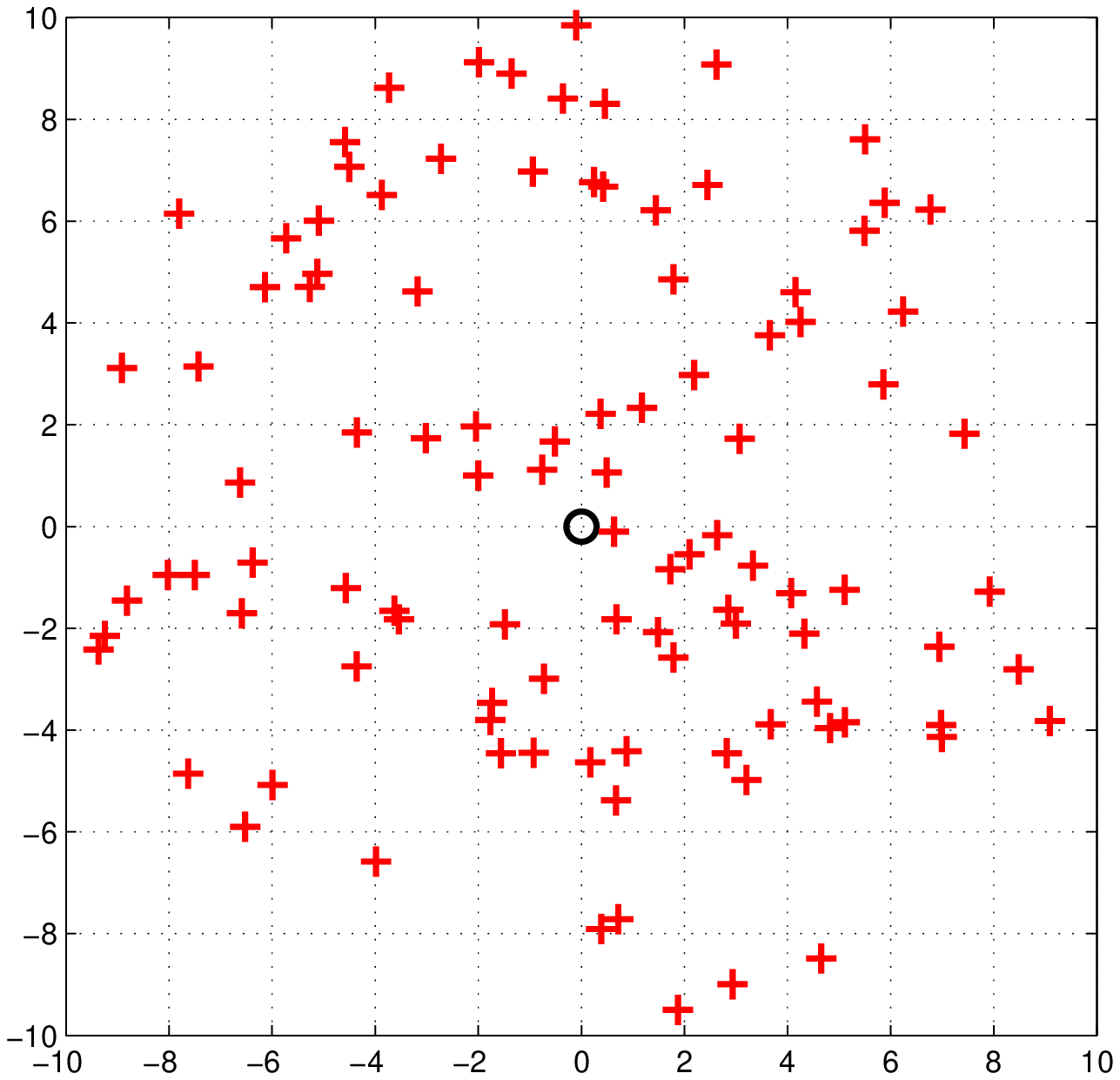}	\\ [-0.2cm]
(a)	& (b)	&	(c)
\end{array}$
\caption{Snapshots of the distribution of ambient RF energy sources (a) $\alpha=-1$ (b) $\alpha=-0.5$ (c) $\alpha=-0.03$.}
\label{Distribution}
\end{center}
\end{figure*} 

\subsection{Performance Metrics}
\label{sec:metrics}

We define the performance metrics of the sensor node as the expectation of RF energy harvesting rate, the variance of the RF energy harvesting rate, power outage probability and transmission outage probability. Let us first introduce the mathematical quantities of interest.

\subsubsection{Theoretical values}
\label{subsec:theoretical}
The expectation of the RF energy harvesting rate is defined as 
$
E_{P_{\mathrm H}} \triangleq \mathbb{E} \left[ P_{\mathrm{H}} \right].
$
 
The variance of RF energy harvesting rate is given by 
$
V_{P_{\mathrm H}} \triangleq\mathbb E\left[\left(P_{\mathrm{H}}-\mathbb{E} \left[ P_{\mathrm{H}} \right]\right)^2\right]. 
$


Power outage occurs when the sensor node becomes inactive due to lack of enough energy supply. The power outage probability is then defined as 
$
	P_{po}	\triangleq	\mathbb{P} \left(	 P_{\mathrm{H}}	<	P_{\mathrm{C}}	\right). 
$

Let $m \ge 0$ denote the minimum transmission rate requirement. If the sensor fails to achieve this requirement, a transmission outage occurs. The transmission outage probability can be defined as 
$
	P_{to} \triangleq	\mathbb{P} \left(	C	<	m	\right). 
$


\subsubsection{Estimation by simulation}

The different theoretical performance metrics introduced in Section~\ref{subsec:theoretical} may in practice be estimated by Monte Carlo simulation of the underlying $\alpha$-DPP. The simulation of $\alpha$-DPPs when $\alpha=-1/j$, $j\in\mathbb N $, is done by using the Schmidt orthogonalization algorithm developed in full generality in~\cite{Hough}, and specifically in~\cite{DecreusefondFlintVergne} for the Ginibre point process. The simple generalization to $\alpha=-1/j$ can be found in the recent survey~\cite{DecreusefondFlintPrivaultTorrisi}, and additional details on DPP can be found in~\cite{DecreusefondFlintPrivaultTorrisi2}.

\subsubsection{Upper bounds under a worst-case scenario}
\label{subsec:worstcase}

In practice, there are no closed forms for all the performance metrics appearing in Section~\ref{subsec:theoretical}. Additionally, estimation by simulation suffers from some drawbacks: 1) the time required to draw $N$ samples of $P_{\mathrm H}$ can be rather long; 2) the estimation may differ significantly from the theoretical result (if $N$ is not sufficiently large);
3)	it is difficult to modify the parameters retroactively, and a new set of simulations is then required.
These drawbacks motivate the computation of upper bounds which is the object of the remainder of the paper.

To that end, we introduce a simpler scenario, which is henceforth called worst-case scenario.
In this worst-case scenario, the sensor node only receives energy from one RF source at a time. In this simpler case, there is energy (respectively transmission) outage if and only if the closest RF source is further than some characteristic distance.
Figure \ref{example} illustrates the difference between the general-case scenario, and the worst-case scenario. $\gamma$ represents the maximum distance from the sensor node where a single RF energy source can still power the sensor node by itself. 
In Fig.~\ref{example1}, as an RF energy source lies in the range of $\gamma$, sufficient power is guaranteed from this single source. As a result, both the general-case scenario and worst-case scenario experience no outage. If there exists no RF energy source in the range of $\gamma$, the sensor may still be powered if the sum of harvested energy from multiple RF energy sources is large enough. In this context, the outage largely depends on the density of RF energy sources. When the density is high, as in Fig.~\ref{example2}, the sensor shall experience no outage in general-case scenarios. However, in the worst-case scenario, 
outage occurs. Contrarily, when the density is low, as in Fig.~\ref{example3} the sensor does not harvest enough energy in both general-case and worse-case scenarios.        


\begin{figure} 
\centering
\subfigure [General Case: No Outage, Worst Case: No Outage] {
 \label{example1}
 \centering
 \includegraphics[width=0.31 \textwidth]{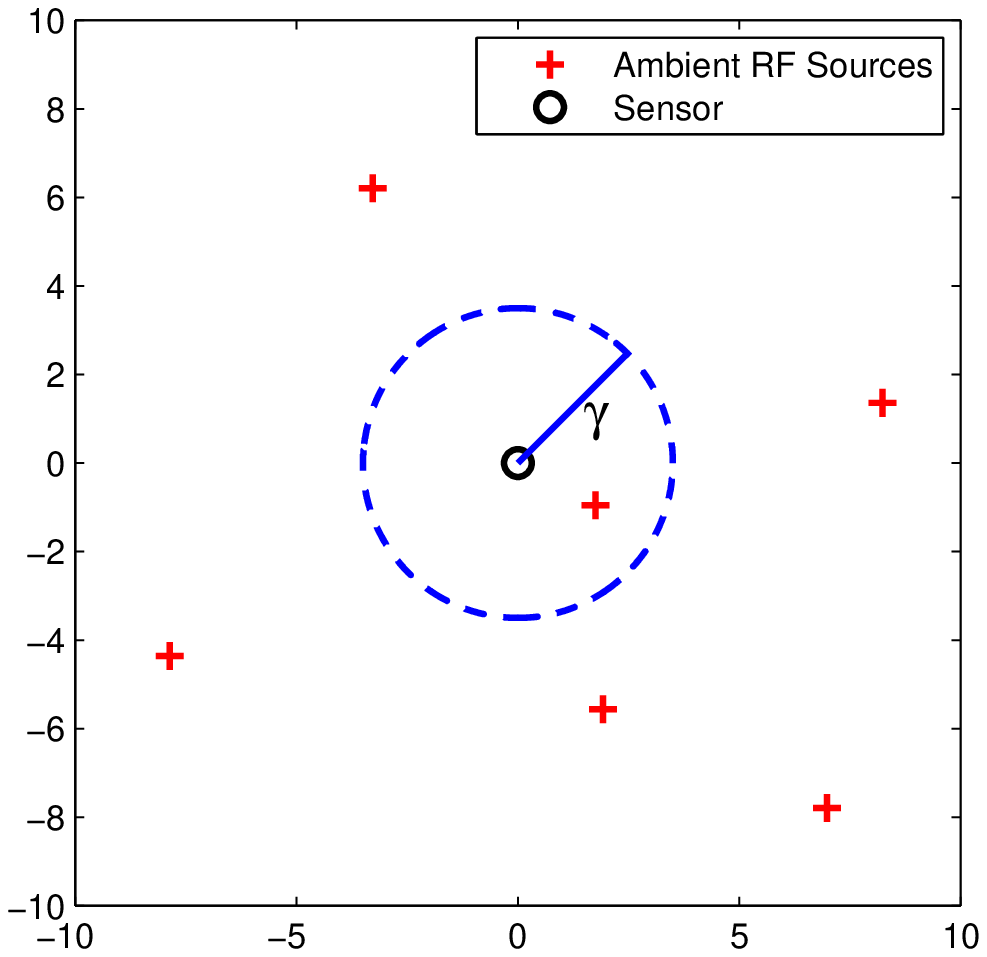}} 
 \centering
 \subfigure [General Case: No Outage, Worst Case: Outage (High RF Source Density)] {
 \label{example2}
    \centering
    \includegraphics[width=0.31 \textwidth]{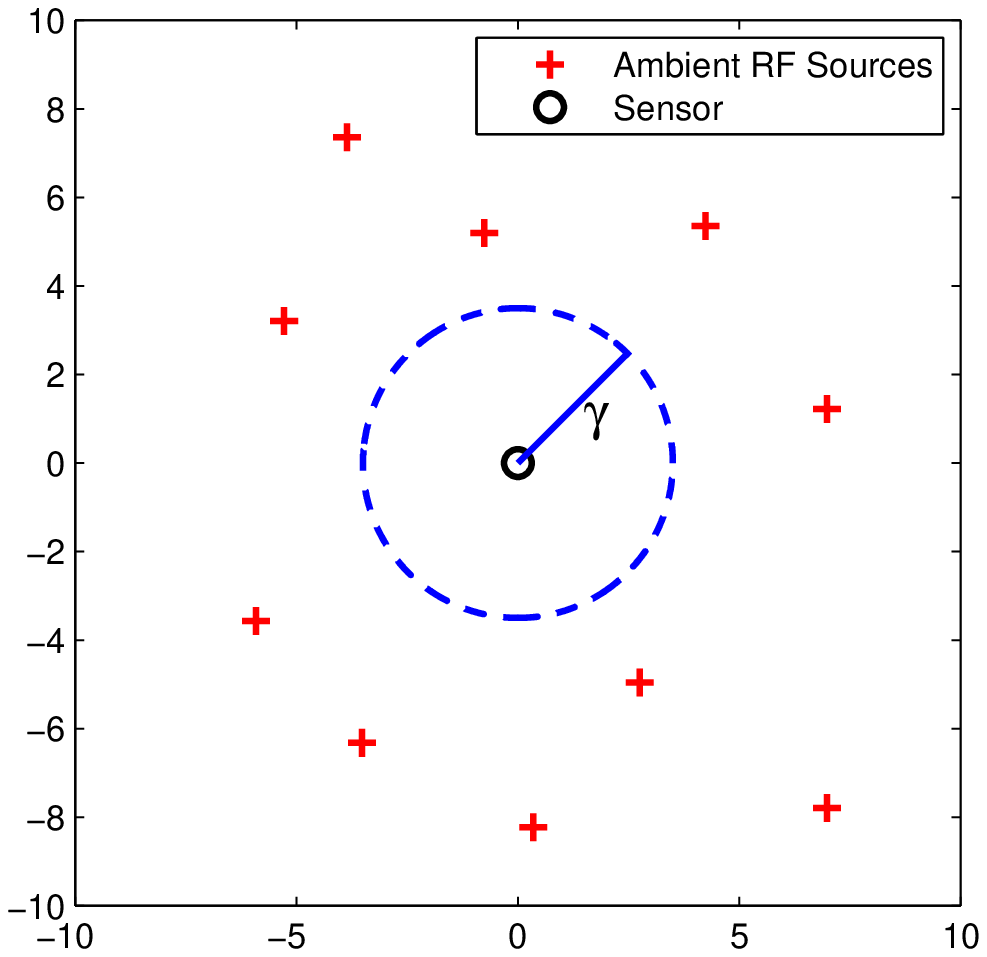}} 
   \centering
 \subfigure [General Case: Outage, Worst Case: Outage (Low RF Source Density)] {
 \label{example3}
    \centering
    \includegraphics[width=0.31 \textwidth]{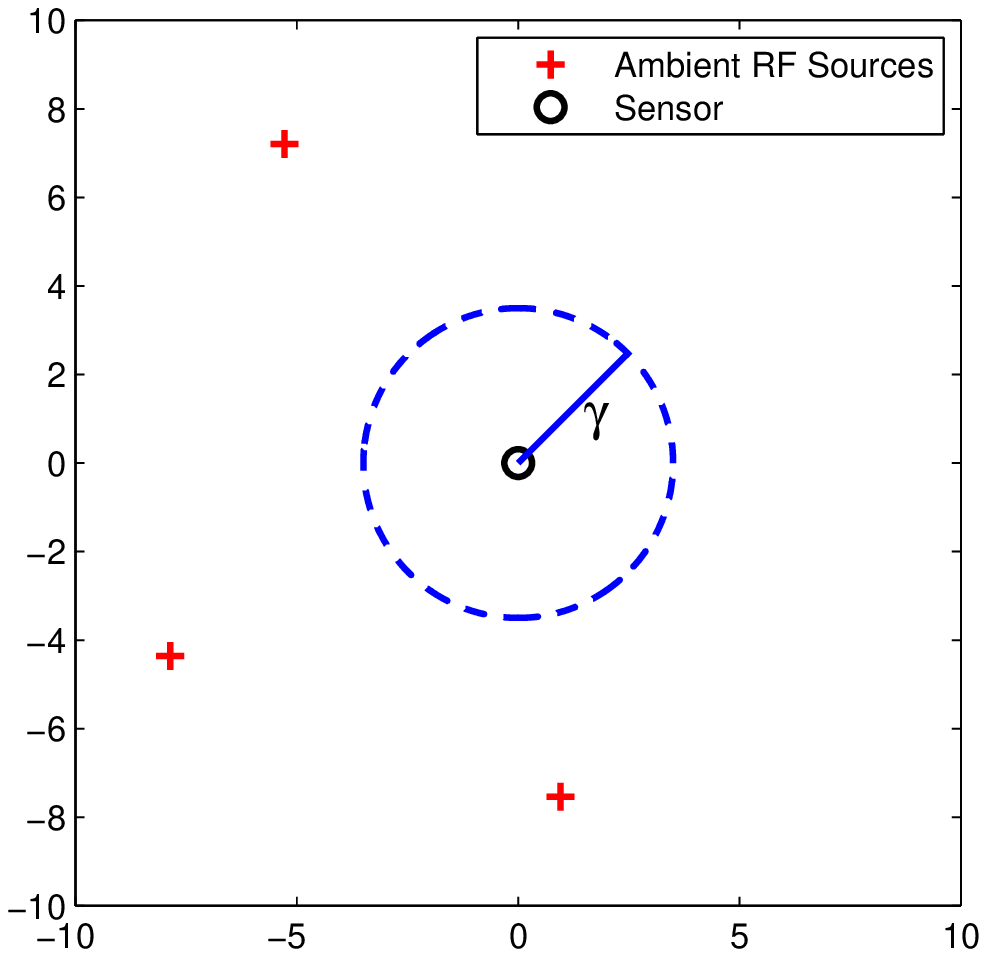}}
    \centering
  \caption{Example scenarios for ambient RF energy harvesting.}
\label{example}
\end{figure}  

As illustrated here, our worst-case scenario over-estimates the outage probability. Indeed in Fig.~\ref{example2}, there is an outage in the worst-case scenario, although there is not in the general scenario. 
Therefore, the performance metrics in this worst-case scenario over-estimate the real performance metrics. this will be the focus of Section~\ref{sec:Analysis} in which we show that this worst-case scenario constitutes an upper-bound to the outage probabilities.

\section{Performance Analysis}
\label{sec:Analysis}

In this section we analyze the performance metrics defined in Section~\ref{sec:metrics} when $\mathcal{K}\sim\mathrm{Gin}(\alpha,\rho)$ is the Ginibre $\alpha$-DPP with parameter $\alpha=-1/j$ for $j\in\mathbb{N}^*$, and density $\rho>0$. 

\subsection{The Expectation and Variance of RF Energy Harvesting Rate}

First, we obtain the expectation of RF energy harvesting rate in the following theorem.
\begin{theorem}
\label{thm:expectedharvestedenergy}
The expectation of RF energy harvesting rate can be explicitly computed as follows\footnote{Here, we say that $f\approx_{\epsilon\rightarrow 0}g$ if $f/g\xrightarrow[\epsilon\to0]{}1$.}:
\begin{align}
\mathbb{E}[P_{\mathrm H}] &=2 \pi\varrho\beta P_{\mathrm{S}} \frac{G_{\mathrm{S}}G_{\mathrm{H}} \lambda^{2}}{(4\pi )^{2}}
\rho\left(\frac{\epsilon}{R+\epsilon}+\ln(R+\epsilon)-1-\ln(\epsilon)\right) \label{eq:average_energy} \\
&\approx_{\epsilon\rightarrow 0} 
\frac{\rho\varrho\beta P_{\mathrm{S}} G_{\mathrm{S}}G_{\mathrm{H}} \lambda^{2}}{8\pi}\ln\left(\frac R\epsilon\right)		.	\label{eq:app_average_energy}
\end{align}
Additionally, the variance of the RF energy harvesting rate can be computed as follows:
\begin{multline}
\label{eq:variance_harvested}
V_{P_H}=\left(\varrho \beta P_{\mathrm{S}} \frac{G_{\mathrm{S}} G_{\mathrm{H}} \lambda^{2}}{(4\pi )^{2}}\right)^2  \\
\left(2\pi\rho\left(\frac{1}{6\epsilon^2}-\frac{3R+\epsilon}{6(R+\epsilon)^3}\right)+\alpha\rho^2\int_{O\times O}\frac{e^{-\pi\rho\lVert \mathbf x-\mathbf y\lVert^2}}{\left(\epsilon+\lVert \mathbf x\lVert\right)^2\left(\epsilon+\lVert \mathbf y\lVert\right)^2}\,\mathrm d\mathbf x\mathrm d\mathbf y\right),
\end{multline}
where recall that $O=\mathcal B(0,R)$ is the observation window.
\end{theorem}

Before moving on to the proof, a few remarks are in order. First, we note that Theorem~\ref{thm:expectedharvestedenergy} implies that at the level of expectations, the Ginibre $\alpha$-DPP behaves like a homogeneous PPP and in particular, the expectation of RF energy harvesting rate is independent of the repulsion parameter $\alpha$. Therefore, on average, the harvested energy is the same when $\alpha$ varies. However, it is straightforward from \eqref{eq:variance_harvested} that the variance of the RF energy harvesting rate is larger when the point process is closer to a PPP. Heuristically, there is a larger probability that there are no points close to the sensor when the RF sources are distributed as a PPP. 
Second, notice that the second term in \eqref{eq:variance_harvested} is in fact not a closed form. To the best of our knowledge, the second term cannot be explicitly calculated, but should instead be approximated numerically.
\begin{proof}
We have
\begin{equation*}
\mathbb{E}[P_{\mathrm H}]=\beta P_{\mathrm{S}} \frac{G_{\mathrm{S}}G_{\mathrm{H}} \lambda^{2}}{(4\pi )^{2}} \int_{O} \frac{\rho^{(1)}(x)}{(\epsilon+\|x\|)^2}\,\mathrm{d}x
\end{equation*}
by Campbell's formula~\cite{Kallenberg}, where 
$\rho^{(1)}(x)=K(x,x)=\rho$ 
 is the intensity function of $\mathcal{K}$ given by \eqref{beq}.
We thus find
\begin{equation*}
\mathbb{E}[P_{\mathrm H}]=\beta  P_{\mathrm{S}} \frac{G_{\mathrm{S}}G_{\mathrm{H}}  \lambda^{2}}{(4\pi )^{2}} 2 \pi \int_{0}^R \rho \frac{r}{(\epsilon+r)^2}\,\mathrm{d}r,
\end{equation*}
 by polar change of variable,
 and the integral on the r.h.s. is computed explicitly as
\begin{equation*}
 \int_{0}^R\frac{r}{(\epsilon+r)^2}\,\mathrm{d}r=
\left(\frac{\epsilon}{R+\epsilon}+\ln(R+\epsilon)-1-\ln(\epsilon)\right),
\end{equation*}
which yields the result.

For brevity, the proof of the expression of variance \eqref{eq:variance_harvested} is presented in {\bf Appendix I}
\end{proof}

\begin{table}
\centering
\caption{\footnotesize Parameter Setting.} \label{parameter_setting}
\begin{tabular}{|l|l|l|l|l|l|l|l|} 
\hline
Symbol & $G^{i}_{S}$, $G^{k}_{S}$ & $\beta$  & $P^{k}_{S}$ & $W$ & $\lambda_{k}$ & $P_{C}$ & $\sigma^2$ \\ 
\hline
Value  & $1.5$ &  $0.3$ & $1W$  & $1KHz$  &  $0.167m$ & $15.8\mu W$ & $-90dBm$ \\
\hline              
\end{tabular}
\end{table}

\begin{figure}
\centering
\includegraphics[width=0.45\textwidth]{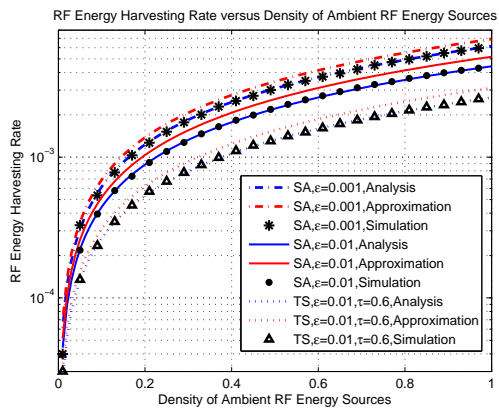}
\caption{RF Energy Harvesting Rate versus Density of Ambient RF Energy Sources. } \label{fig:amount_of_energy}
\end{figure}

Next, we examine the validity of the expressions of the expectation and variance of RF energy harvesting rate. All the network simulations in this paper are considered in the scenario of an LTE network with a typical 1800$MHz$ operating frequency. The corresponding wave length adopted is 0.167$m$. 
The channel gain of both transmit antenna and receive antenna are assumed to be 1.5. The RF-to-DC power conversion efficiency is assumed to be 30$\%$. We consider the transmit power of ambient RF sources is 1$W$. The circuit power consumption of the sensor is fixed to be -18$dBm$ (i.e., 15.8$\mu W$) as in~\cite{N.2013Parks}.
The sensor is assumed to be allocated with a 1$kHz$ bandwidth for data transmission. The AWGN is considered to be -90$dBm$.
The channel gain between the sensor and data sink is calculated as $h_{0}=62.5d^{-4}$~\cite{2014X.Lu}, where $d$ is the distance between the sensor node and the data sink. The results with the separated and time-switching receiver architectures are labeled as ``SA" and ``TS", respectively. Note that the results for the PPP are identical to that of the $\alpha$-DPP, when $\alpha=0$. Additionally, the performance of separated receiver architecture in terms of expectation and variance of RF energy harvesting rate and power outage probability is identical to the case when $\tau=1$ for time-switching. 

As shown in Fig. \ref{fig:amount_of_energy}, the numerical results, averaged over $N=5\times 10^5$ of simulation runs, match the analytical expression (\ref{eq:average_energy}) accurately over a wide range of density $\rho$, i.e., from $0.01$ to $1$. As expected, the RF energy harvesting rate increases with the density of ambient RF energy sources. Under the same density, the RF energy harvesting rate is affected by $\epsilon$. We observe that when $\epsilon=0.001$, larger RF energy harvesting rate is available at the sensor than that when $\epsilon=0.01$. The straightforward reason is that, from \eqref{eq:totalamountofpower}, the smaller distance the RF sources locate near the sensor, the more aggregated RF energy harvesting rate can be achieved. Additionally, the use of the approximate expression (\ref{eq:app_average_energy}) from (\ref{eq:average_energy}) can be observed to increase with the density of ambient RF energy sources. However, the difference between (\ref{eq:app_average_energy}) and (\ref{eq:average_energy}) in percentage remains the same, i.e., 15 percent when $\epsilon=0.01$. This difference is dependent on $\epsilon$, and is shown to diminish with the decrease of $\epsilon$. As shown in Fig. \ref{fig:amount_of_energy}, when $\epsilon=0.001$, the approximate expression approximates more closely to the analytical expression.  

\subsection{Upper Bound of the Power Outage Probability}
 \label{sec:upperboundpoutage}
 
In this section, we derive upper-bounds of power outage probability defined in Section~\ref{subsec:theoretical}. We interpret these upper-bounds in terms of a worst-case scenario, as specified in Section~\ref{subsec:worstcase}. We also point out that all the numerical estimations done in this section are performed under the worst-case assumption. 

\begin{theorem}
\label{thm:estimationphi}
 Let us define 
\begin{equation*}
\gamma\triangleq \frac\lambda{4\pi}\sqrt{ \frac{\varrho\beta P_{\mathrm{S}}G_{\mathrm{S}} G_{\mathrm{H}} }{P_{\mathrm{C}} } }.
\end{equation*}
Then, the following bound holds:
\begin{eqnarray} \label{DPP_poweroutage}
P_{po} = \mathbb{P} (P_H<P_C)
\le  \left(\prod_{n\ge 0} \left(1+\alpha  \frac{\Gamma(n+1, \pi\rho\inf(R,\gamma)^2)}{n!}\right)\right)^{-1/\alpha},
\end{eqnarray}
where $\Gamma(z,a)$ is the lower incomplete Gamma function defined in \eqref{eq:defgamma}.
\end{theorem}
Remark that the parameter $\epsilon$ does not appear in the bound of Theorem~\ref{thm:estimationphi} and rigorously, the inequality holds only for a sufficiently small value of $\epsilon$. In practice, we should therefore make sure that $\epsilon$ is chosen small enough.
\begin{proof}
To make the proof easier to follow, let us set
$
f(\mathbf{x}_k)\triangleq \varrho\beta P_{\mathrm{S}} \frac{G_{\mathrm{S}} G_{\mathrm{H}} \lambda^{2}}{(4\pi (\epsilon+\|\mathbf{x}_k\|))^{2}},
$
for $k\in\mathcal{K}$. Then,
\begin{align}
\label{eq:upperboundinthm2}
P_{po}	 
	=\mathbb{P} \left(\sum_{k\in\mathcal{K}} f(\mathbf{x}_k) \le P_{\mathrm{C}}\right) \le \mathbb{P} (\forall k \in\mathcal{K},\ f(\mathbf{x}_k) \le P_{\mathrm{C}})
	&= \mathbb{P} (\forall k \in\mathcal{K},\ \|\mathbf{x}_k\| \ge \gamma -\epsilon ) \\&= \mathbb{P} (\mathcal{K}\cap\mathcal{B}(0,\gamma-\epsilon) = \emptyset ),\nonumber
\end{align}
where we have chosen $\epsilon$ such that $\gamma-\epsilon\ge 0$. Thus by Proposition~\ref{prop:holeproba}, we obtain
\begin{eqnarray}\label{det}
P_{po}  \le \mathrm{Det}(\mathrm{Id}+\alpha K_{\mathcal{B}(0,\gamma-\epsilon)})^{-1/\alpha}.
\end{eqnarray}
Since in our case $K$ is the Ginibre kernel, the eigenvalues of $K$ are given by \eqref{eq:eigenvalues}. By standard properties of the Fredholm determinant which can be found, e.g., in~\cite{Brezis}, we find,
\begin{equation*}
P_{po}  
\le  \left(\prod_{n\ge 0} \left(1+\alpha  \frac{\Gamma(n+1, \pi\rho \inf(R,\gamma-\epsilon)^2)}{n!}\right)\right)^{-1/\alpha},
\end{equation*}
and the result follows by letting $\epsilon$ go to zero on the r.h.s. of \eqref{det}, since the associated function of $\epsilon$ is continuous.
\end{proof}

Let us explain briefly the heuristics behind the bound in Theorem~\ref{thm:estimationphi}. Note that the only line which is not an equality in the proof of Theorem~\ref{thm:estimationphi} is \eqref{eq:upperboundinthm2}. The approximation made therein is that the harvested energy rate $P_{\mathrm H}$ is provided by the RF source closest to the sensor node. Hence, the probability appearing on the r.h.s. of \eqref {DPP_poweroutage} is the probability corresponding to the worst-case scenario introduced in Section~\ref{subsec:worstcase}. Estimation of this upper-bound by simulation will therefore consist in assuming the worst-case scenario. Conversely, estimation of the l.h.s. of \eqref{DPP_poweroutage} will consist in assuming the general-case scenario.

It should be noted that the eigenvalues appearing in the product of Theorem~\ref{thm:estimationphi} are in decreasing order, and decrease exponentially when $n \ge  \pi\rho \inf(R,\gamma)^2$, see~\cite{DecreusefondFlintVergne} for details. Hence, the product which appears in Theorem~\ref{thm:estimationphi} is well approximated by 
\begin{eqnarray}
 \left(\prod_{n\ge 0}^N \left(1+\alpha  \frac{\Gamma(n+1,  \pi\rho\inf(R,\gamma)^2)}{n!}\right)\right)^{-1/\alpha},
\end{eqnarray}
 where $N \gg \pi\rho\inf(R,\gamma)^2$.
We also note that
\begin{eqnarray}
  \frac{\mathrm{d}}{\mathrm{d}\alpha} \ln\left(\prod_{n\ge 0} (1+\alpha \lambda_n)\right)^{-1/\alpha} 
 	=   \frac{1 }{\alpha^2}  \sum_{n\ge 0} \frac{(1+\alpha \lambda_n)\ln(1+\alpha \lambda_n) - \alpha \lambda_n}{1+\alpha \lambda_n} \ge 0,  \nonumber
\end{eqnarray}
which means that the bound of Theorem~\ref{thm:estimationphi} is smallest when $\alpha=-1$, i.e. when repulsion is maximal, and increases with $\alpha$. 

As a corollary of Theorem~\ref{thm:estimationphi}, we find in the case of a PPP (which is obtained as the limit as $\alpha\rightarrow 0$ in the theorem) the next corollary.
\begin{corollary}
Let $\mathcal{K}\sim\mathrm{Poiss}(O,\rho)$ be a Poisson process on $O=\mathcal{B}(0,R)$ with density $\rho$. Then, the following bound holds:
\begin{eqnarray} 
\label{PPP_Poweroutage}
P_{po} \le e^{-\pi\rho\inf(R,\gamma)^2},  
\end{eqnarray}
where $\gamma$ is as defined in Theorem~\ref{thm:estimationphi}.
\end{corollary}

 \begin{figure}
 \centering
  \begin{minipage}[c]{0.48\textwidth}
          \includegraphics[width=0.95\textwidth]{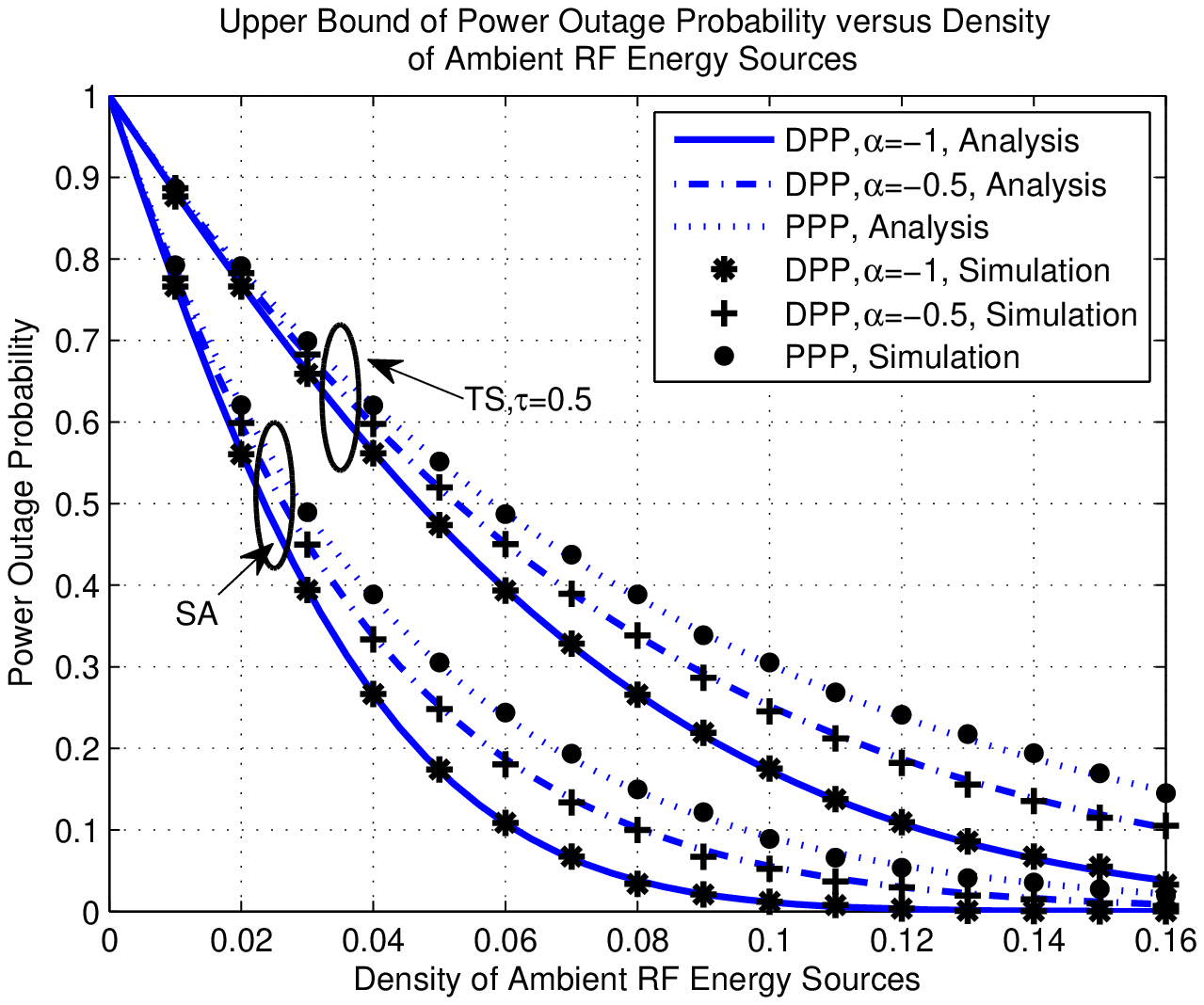}
           \caption{Upper bound of power outage probability versus density of ambient RF energy source (separated receiver architecture).}  \label{fig:power_outage}          
    \end{minipage}
  \begin{minipage}[c]{0.48\textwidth}
          \includegraphics[width=0.95\textwidth]{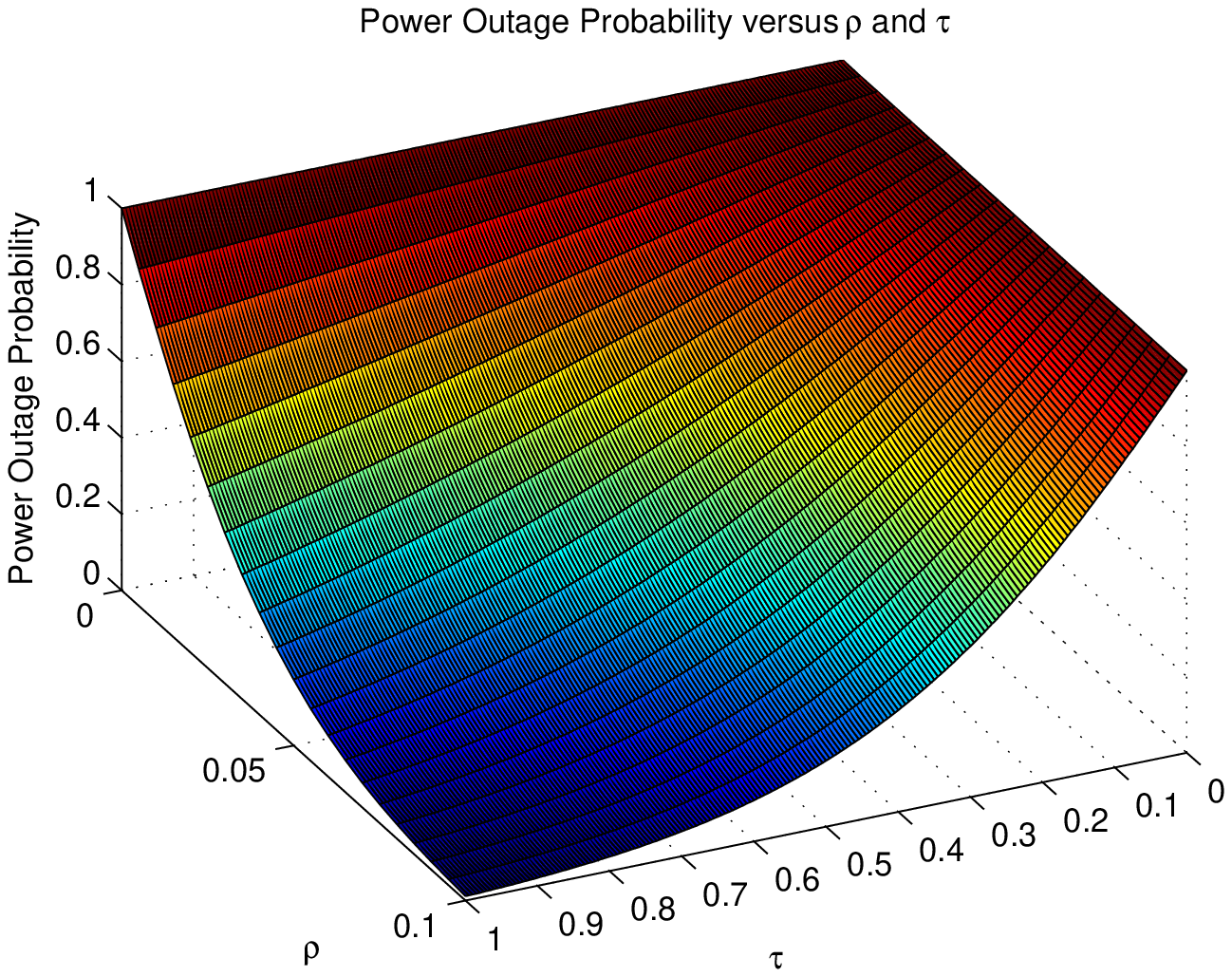}
            \caption{Upper Bound of Power Outage Probability versus Density of Ambient RF  Energy Sources and Time-Switching Coefficient $\tau$ (time-switching architecture).}\label{fig:power_outage_TS}        
  \end{minipage}
  \end{figure}



In Fig. \ref{fig:power_outage}, we show the variation of the upper bound of the power outage probability $P_{po}$ as a function the density of ambient RF energy sources $\rho$ for both separated and time-switching architecture. It is observed that $P_{po}$ is a decreasing function of $\rho$. 
The numerical results verify that the analytical expressions for the upper bounds in \eqref{DPP_poweroutage} and \eqref{PPP_Poweroutage} are very tight for different values of $\alpha$. We also observe that the smaller the $\alpha$, the lower the upper bound of the outage probability. In other words, more repulsion among the locations of the RF sources results in lower outage probability of the sensor. This finding coincides with the previous analysis that the maximal repulsion results in best power outage performance. 
Moreover, we can observe that the influence of $\alpha$ depends on the density $\rho$. For both receiver architectures, the gap between the bounds for DPP ($\alpha=-1$) and PPP first increases with the density $\rho$ and then decreases when the power outage probability becomes low.

Figure~\ref{fig:power_outage_TS} examines the impact of time-switching coefficient $\tau$ on the upper bound of the outage probability. It is obvious that, regardless of density $\rho$, power outage probability is a monotonically decreasing function of $\tau$. 
As is seen on Figure~\ref{fig:power_outage_TS}, a small value of $\tau$ can significantly degrade the upper bound of power outage performance. The more the value of $\tau$ decreases, the more quickly the power outage performance degrades.

Then, we evaluate the impact of circuit power consumption $P_C$ of the sensor for separate and time-switching receiver architectures in Fig. \ref{fig:power_outage_SA} and Fig. \ref{fig:power_outage_PC_tau}, respectively. It is seen that, when the density $\rho$ is small (e.g., $\rho=0.01$), the corresponding plot is a logarithm-like function. Specifically, the power outage probability is very sensitive to the small value of $P_C$, and becomes less sensitive when $P_C$ is larger (e.g., above $2\times 10^{-5}$W). For example, the upper bound of power outage probability for PPP increases from $15.1\%$ to $58.4\%$ (i.e., $43.3\%$ difference) when $P_C$ varies from $0.2\times 10^{-5}$W to $0.7\times 10^{-5}$W. However, when $P_C$ varies in the same amount from $2.0\times 10^{-5}$W to $2.5\times 10^{-5}$W, the upper bound only changes from $82.9\%$ to $86.1\%$ (i.e., $3.2\%$ difference). This indicates that, in practice, advanced sensor circuit with small $P_C$ can help to lower down power outage probability significantly, especially when the available RF energy harvesting rate is small. For the scenarios with larger density $\rho$ (e.g., $\rho=0.03$), the outage probability tends to grow smoothly with the increase of $P_C$. The gap between the upper bounds of the DPP ($\alpha=-1$) and the PPP increases with the density $\rho$. This can be understood since a larger number of random RF sources results in a larger variance in RF energy harvesting rate (shown in \eqref{eq:variance_harvested} in {\bf Theorem 1}), thus a greater difference in the power outage probability. For time-switching, as shown in Fig. \ref{fig:power_outage_PC_tau}, the impact of $P_C$ on the power outage probability is dependent on the coefficient $\tau$. The larger $\tau$ is, the more dynamically the upper bound of power outage probability varies when $P_C$ is small. 

Lastly, we emphasize that in practice, the computation of the upper-bound obtained in Theorem~\ref{thm:estimationphi} is very fast. Therefore, we believe that these bounds are the preferred choice in practical applications.

 \begin{figure}
 \centering
  \begin{minipage}[c]{0.48\textwidth}
          \includegraphics[width=0.95\textwidth]{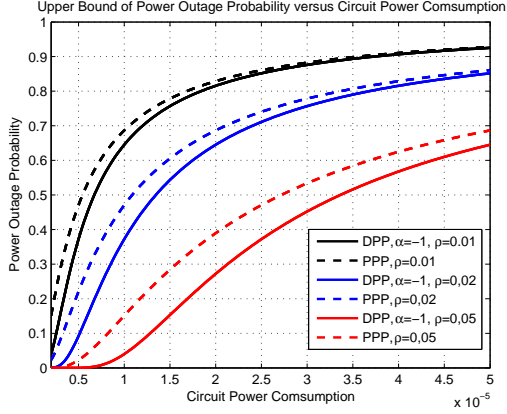}
           \caption{Upper bound of power outage probability versus density of ambient RF energy source (separated receiver architecture).}  \label{fig:power_outage_SA}          
    \end{minipage}
  \begin{minipage}[c]{0.48\textwidth}
          \includegraphics[width=0.95\textwidth]{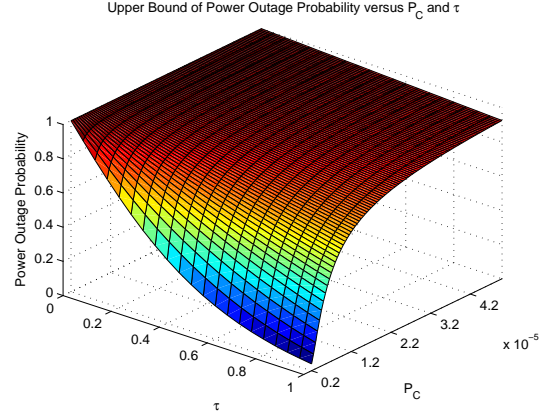}
            \caption{Upper Bound of Power Outage Probability versus Circuit Power Consumption and Time-Switching Coefficient $\tau$ (time-switching architecture).}   \label{fig:power_outage_PC_tau}     
  \end{minipage}
  \end{figure}



\subsection{Upper Bound of the Transmission Outage Probability}

Next, we give a practical upper bound for the estimation of the transmission outage probability $P_{to}$. Again, we interpret these upper-bounds in terms of a worst-case scenario, which was specified in Section~\ref{subsec:worstcase}. 
\begin{theorem}
\label{thm:transmissionoutage}
Let us define:
\begin{equation}
\label{eq:gammam}
	\gamma_m\triangleq \frac\lambda{4\pi}\sqrt{ \frac{P_{\mathrm S}\varrho\beta G_{\mathrm{S}} G_{\mathrm{H}}\left(h_0+\eta\xi\left(1-2^{m/(\eta W)}\right)\right)}{P_{\mathrm C}h_0+\eta\sigma^2\left(2^{m/(\eta W)}-1\right)}}.
\end{equation}
Then we obtain
\begin{equation}
\label{eq:boundpsi}
	P_{to} = \mathbb{P} \left(	C	<	m	\right)
	\le  \left(\prod_{n\ge 0} \left(1+\alpha  \frac{\Gamma(n+1,  \pi\rho\inf(R,\gamma_m)^2)}{n!}\right)\right)^{-1/\alpha}.
\end{equation}
\end{theorem}
\begin{proof}
Applying the definition of $C$ given in \eqref{eq:maxtransmission}, it holds that
\begin{align*}
	P_{to} &= \mathbb{P} \left( \frac{h_0}\eta\left[	P_{\mathrm{H}} - P_{\mathrm{C}} \right]^+<\left(\sigma^2+\xi \sum_{k\in\mathcal K}P_{\mathrm H}^k\right)\left(2^{m/(\eta W)}-1\right) \right)\nonumber\\
	&= \mathbb{P} \left(P_{\mathrm{H}}<P_{\mathrm C}+\eta\left(\sigma^2+\xi P_{\mathrm H}\right)\frac{2^{m/(\eta W)}-1}{h_0} \right)\nonumber\\
	&=\mathbb{P} \left(P_{\mathrm{H}}\left(h_0-\eta\xi\left(2^{m/(\eta W)}-1\right)\right)<h_0P_{\mathrm C}+\eta\sigma^2\left(2^{m/(\eta W)}-1\right)\right)\nonumber\\
	\end{align*}
where we have used the fact that $2^{m/(\eta W)}-1\ge 0$. This implies that if $h_0-\xi\left(2^{m/(\eta W)}-1\right)\le 0$, $P_{to} =1$. Now, if $h_0-\xi\left(2^{m/(\eta W)}-1\right)>0$, it remains to reorganize the equation in order to use Theorem~\ref{thm:estimationphi}:
\begin{equation}
\label{eq:intermeqthm3}
	P_{to}=\mathbb{P} \left(P_{\mathrm{H}}<\frac{h_0P_{\mathrm C}+\eta\sigma^2\left(2^{m/(\eta W)}-1\right)}{h_0-\eta\xi\left(2^{m/(\eta W)}-1\right)}\right),
\end{equation}
so we conclude by Theorem~\ref{thm:estimationphi} that 
$
	P_{to}  \le\left(\prod_{n\ge 0} \left(1+\alpha  \frac{\Gamma(n+1,  \pi\rho\inf(R,\gamma_m)^2)}{n!}\right)\right)^{-1/\alpha},
$
where 
\begin{align*}
	\gamma_m=\frac\lambda{4\pi}\sqrt{ \frac{\varrho\beta P_{\mathrm{S}}G_{\mathrm{S}} G_{\mathrm{H}} }{\left(\frac{h_0P_{\mathrm C}+\eta\sigma^2\left(2^{m/(\eta W)}-1\right)}{h_0-\eta\xi\left(2^{m/(\eta W)}-1\right)}\right) } }  
		=\frac\lambda{4\pi}\sqrt{ \frac{P_{\mathrm S}\varrho\beta G_{\mathrm{S}} G_{\mathrm{H}}\left(h_0+\eta\xi\left(1-2^{m/(\eta W)}\right)\right)}{P_{\mathrm C}h_0+\eta\sigma^2\left(2^{m/(\eta W)}-1\right)}}.
\end{align*}
\end{proof}  

We note that equation \eqref{eq:intermeqthm3} is an equality, so the only inequality in the proof of Theorem~\ref{thm:transmissionoutage} traces back to the application of Theorem~\ref{thm:estimationphi}. So similar to the observations made after the proof of Theorem~\ref{thm:estimationphi}, we note that the probability appearing on the r.h.s. of \eqref {eq:boundpsi} corresponds to the worst-case scenario introduced in Section~\ref{subsec:worstcase}. Estimation of the upper-bound by simulation will therefore consist in assuming the worst-case scenario.
  
\begin{figure} 
\centering
\subfigure [Out-of-band Transmission ($d=50$, $m=3$)] {
\label{Transoutage_out}
\centering
\includegraphics[width=0.45 \textwidth]{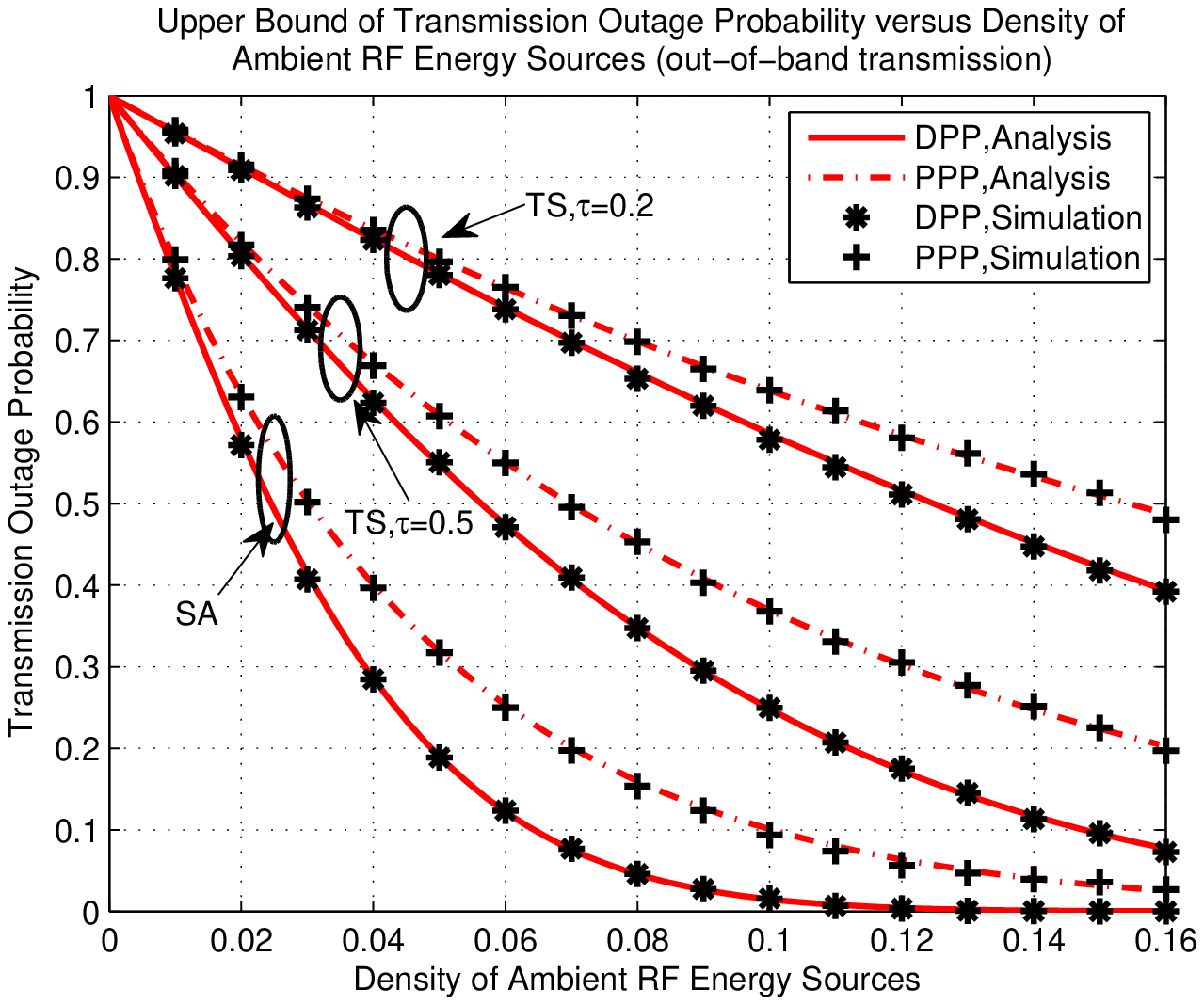}} 
\centering
\subfigure [In-band Transmission ($d=5$, $m=0.02$)] {
\label{Transoutage_in}
\centering
\includegraphics[width=0.45 \textwidth]{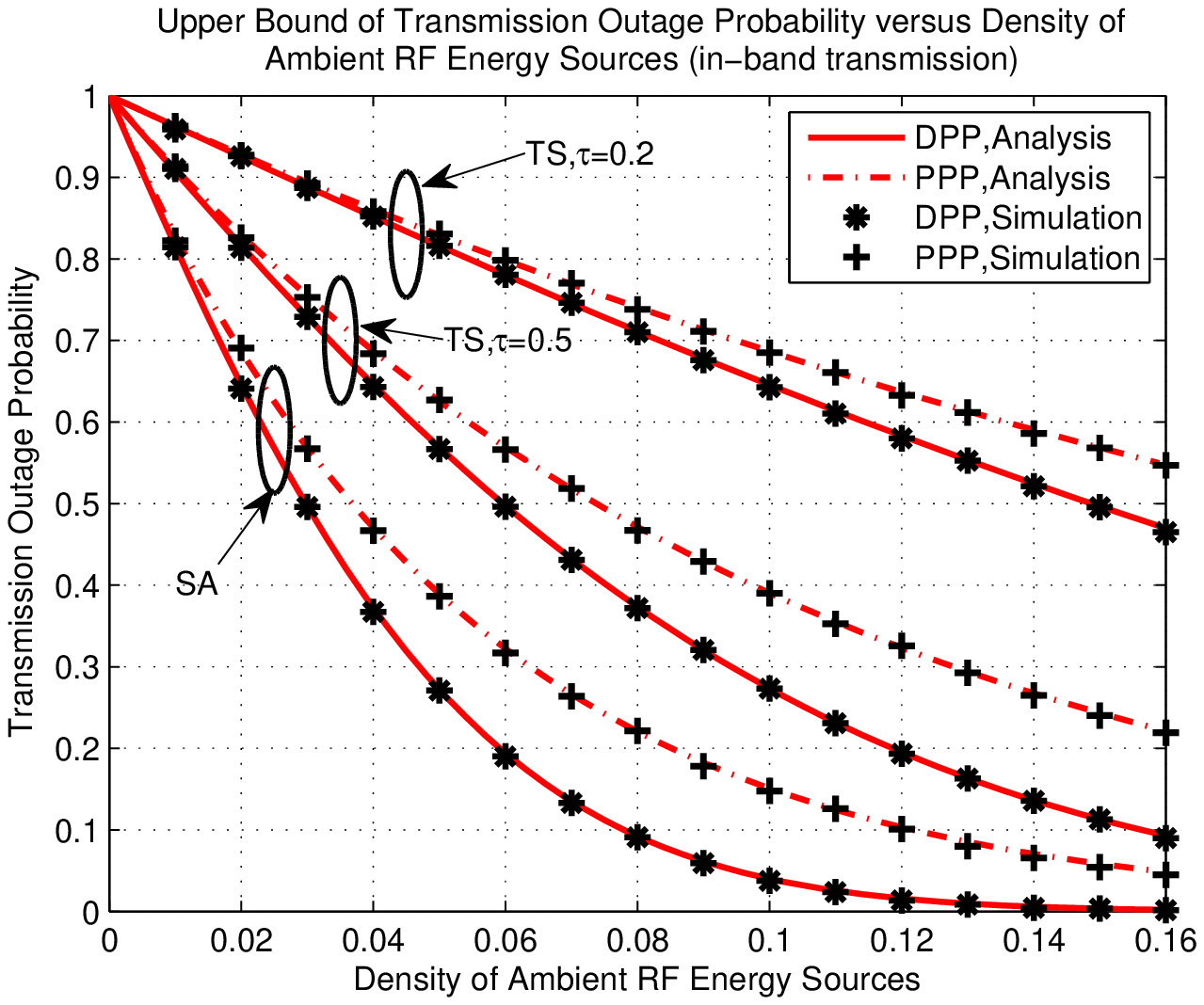}} \\
\centering
\caption{Upper bound of transmission outage probability versus density of ambient RF energy sources.}
\label{fig:trans_outage}
\end{figure}

\begin{figure}
\centering
\includegraphics[width=0.45\textwidth]{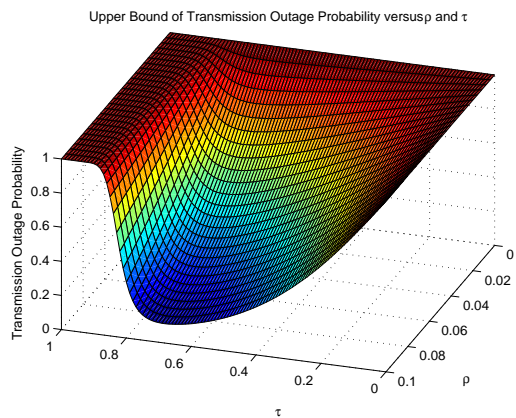}
\caption{Upper Bound of Transmission Outage Probability versus Density of Ambient RF Energy Sources and Time-Switching Coefficient $\tau$ (out-of-band transmission).} \label{fig:Transoutage_TS_out}
\end{figure}

We may also estimate the transmission outage probability when $\mathcal{K}$ is a Poisson process as follows.
\begin{corollary}
Let $\mathcal{K}\sim\mathrm{Poiss}(O,\rho)$ be a Poisson process on $O=\mathcal{B}(0,R)$ with density $\rho$. Then, the following bound holds:
\begin{eqnarray}
	P_{to} \le e^{-\pi\rho\inf(R,\gamma_m)^2},
\end{eqnarray}
where $\gamma_m$ is defined in \eqref{eq:gammam}. 
\end{corollary}
 
In Fig. \ref{fig:trans_outage}, we evaluate the analytical expressions of the upper bound of transmission outage probability for both scenarios of out-of-band and in-band transmission. For the former and latter scenarios, we set the distance between the sensor and data sink $d$ to be 50m and 5m, and the transmission rate requirement to be 3kbps and 0.02kbps, respectively. The numerical results are averaged over $10^{6}$ runs of simulation. Similar to the upper bound of power outage probability, the upper bound of transmission outage probability demonstrates a similar pattern, i.e., a decreasing function of density $\rho$. It is shown that the analytical upper bounds for both receiver architectures in both scenarios of out-of-band and in-band transmission are very accurate.
 
For the time-switching architecture, we examine the impact of the coefficient $\tau$ in Fig. \ref{fig:Transoutage_TS_out}. We can observe that the larger the density $\rho$ is, the more the value of $\tau$ affects the transmission outage performance. Unlike power outage probability which is a decreasing function of $\tau$, transmission outage probability has a minimal value attained for $\tau\in(0,1)$. For a certain density $\rho>0$, when $\tau$ varies from $0$ to $1$, the upper bound of transmission outage probability first decreases, but begins to increase quickly after $\tau$ exceeds a certain value. Then, after $\tau$ reaches another certain threshold, the upper bound remains to be $1$. This implies that there exists some tradeoff between the energy harvesting time and data transmission time to minimize transmission outage probability. Another interesting finding in Fig. \ref{fig:Transoutage_TS_out} is that the optimal value of $\tau$ to minimize the transmission outage probability are not dependent on the density $\rho$. This will be proved in Section~\ref{subsec:optimaltau}.

\begin{figure} 
\centering
\subfigure [Out-of-band Transmission ($d=50$)] {
\label{Transoutage_SA_TS_rate_out}
\centering
\includegraphics[width=0.45 \textwidth]{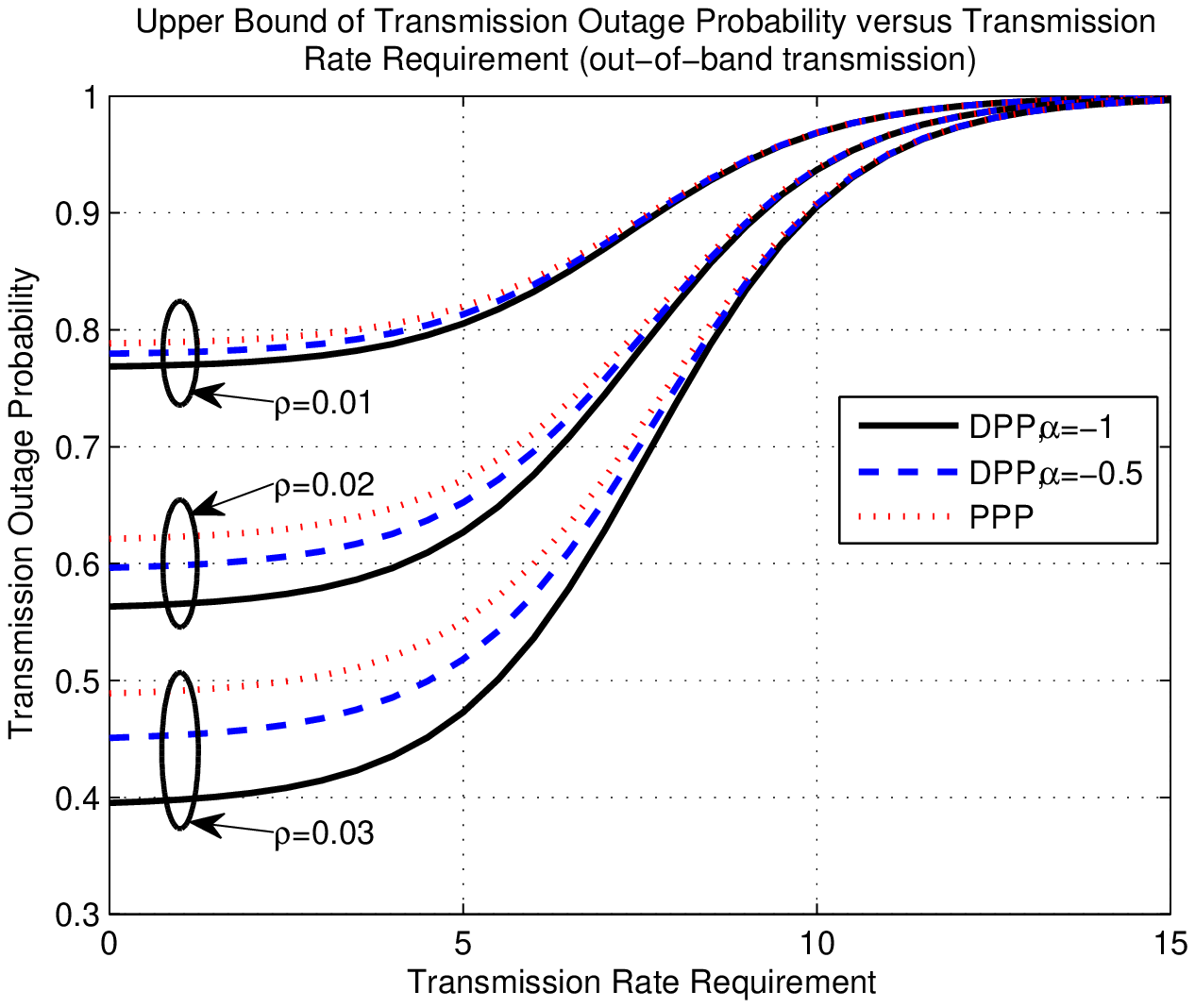}} 
\centering
\subfigure [In-band Transmission ($d=5$)] {
\label{Transoutage_SA_TS_rate_in}
\centering
\includegraphics[width=0.45 \textwidth]{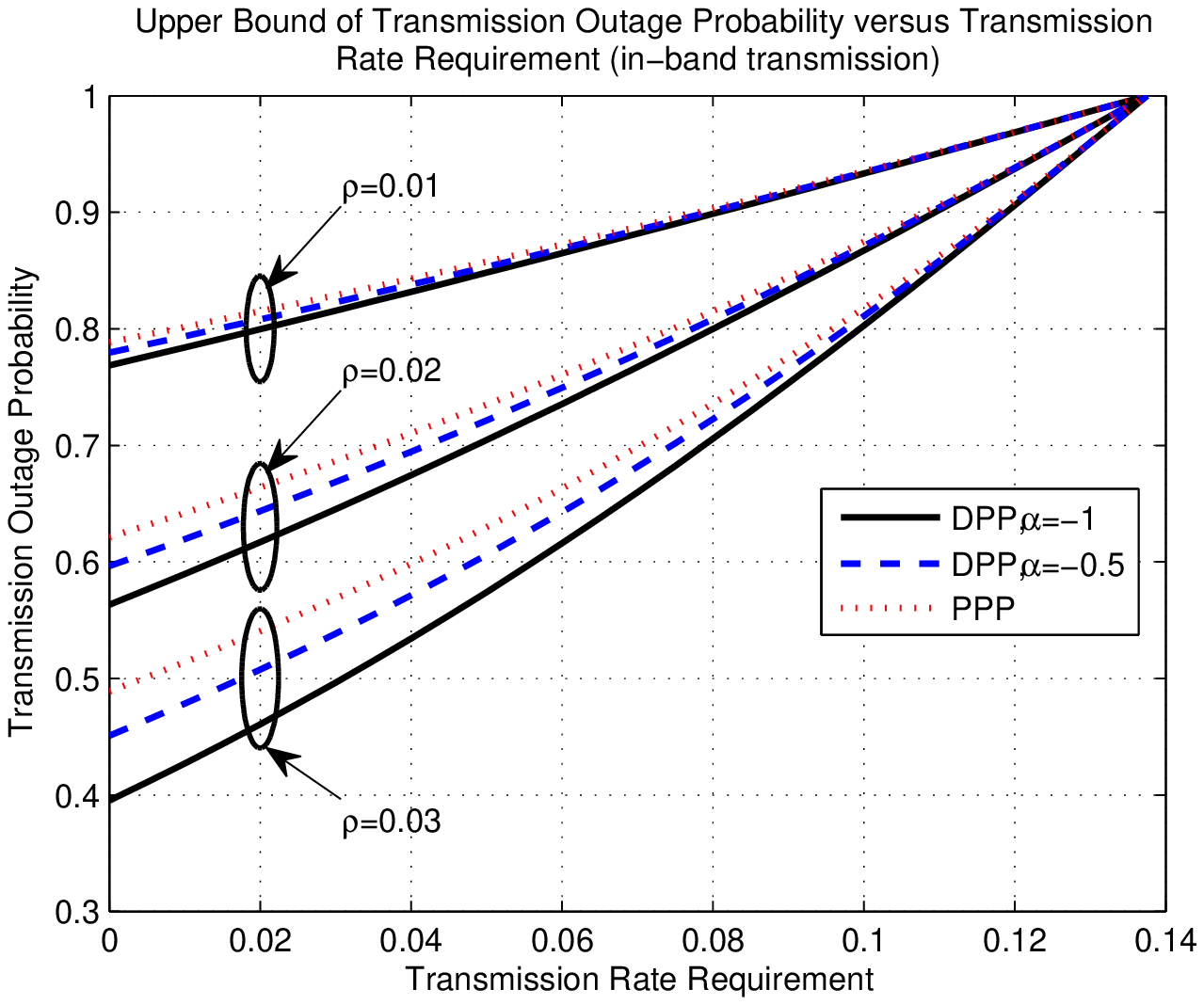}} \\
\centering
\caption{Upper bound of transmission outage probability versus transmission rate requirement (separated receiver architecture).}
\label{fig:transoutage_rate}
\end{figure}

\begin{figure} 
\centering
\subfigure [Out-of-band Transmission ($d=50$, $\rho=0.03$ )] {
\label{transoutage_rate_TS_out}
\centering
\includegraphics[width=0.45 \textwidth]{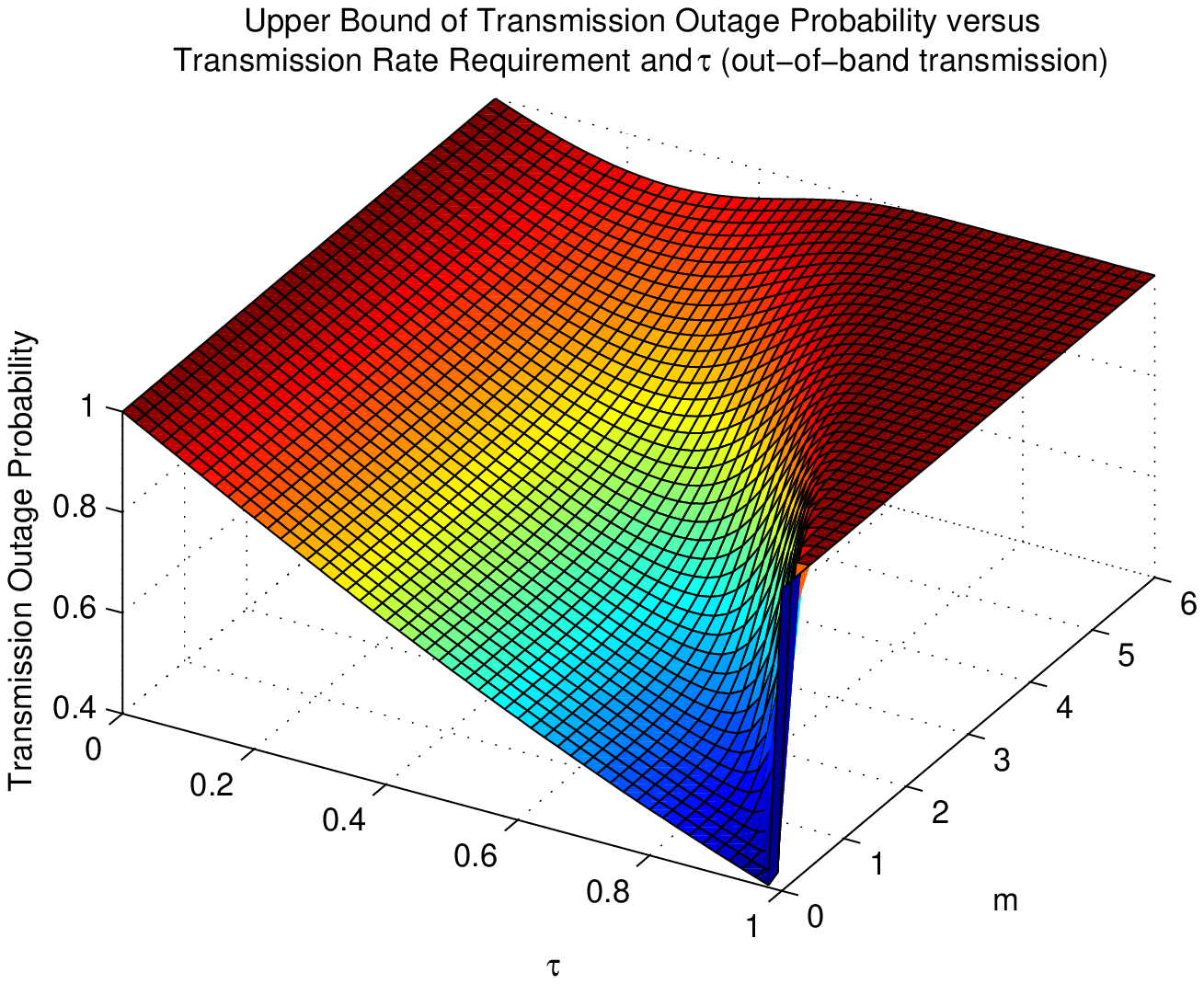}} 
\centering
\subfigure [In-band Transmission ($d=5$, $\rho=0.03$ )] {
\label{transoutage_rate_TS_in}
\centering
\includegraphics[width=0.45 \textwidth]{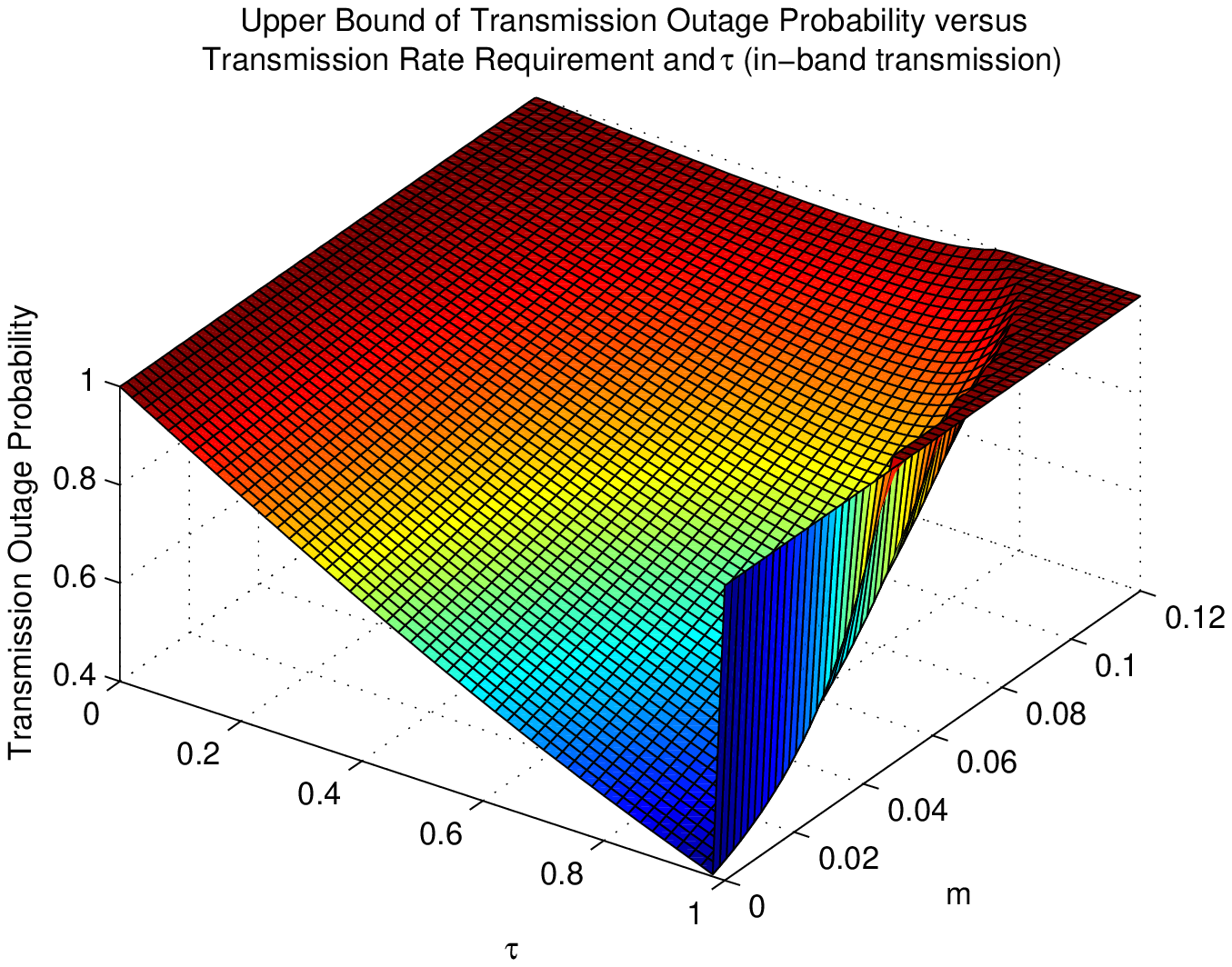}} \\
\centering
\caption{Upper bound of transmission outage probability versus transmission rate requirement and $\tau$ (time-switching architecture).}
\label{fig:trans_outage_alpha}
\end{figure}

Next, in Fig. \ref{fig:transoutage_rate} and Fig. \ref{fig:trans_outage_alpha}, we illustrate the impact of transmission rate requirement on the upper bounds of transmission outage probability for separated and time-switching receiver architecture, respectively. It is clear from the figures that the transmission outage probability is an increasing function of transmission rate requirement for both architectures.
As can be seen from Fig. \ref{fig:transoutage_rate}, 
a larger density $\rho$ causes a bigger difference between the upper bounds of transmission outage probability for the DPP ($\alpha=-1$) and the PPP. This is because when the density $\rho$ is small, there are few RF sources, and their correlation structure does not have such a high impact.
Similar to the study in power outage probability, the difference when $\alpha$ varies is caused by the increased variance of RF energy harvesting rate. Indeed, since the variance is higher, when $\alpha$ is close to $-1$, the transmission outage probability is also higher.
For out-of-band transmission, the upper bound of the transmission outage probability is a sigmoid shape function of transmission rate requirement. Specifically, when the transmission rate requirement is relatively small or high, the upper bound of transmission outage probability increases more slowly with the transmission rate requirement, compared to the case when the transmission rate requirement is median. 
By contrast, for in-band transmission, the transmission outage probability is close to a linear function of rate requirement. In Fig.~\ref{fig:trans_outage_alpha}, for both scenarios of out-of-band transmission and in-band transmission, we can observe that the optimal $\tau$ is dependent on the rate requirement. Generally, the larger the rate requirement, the smaller the value of optimal $\tau$. This motivates us to study the optimal value of $\tau$ in the next subsection.

Based on {\bf Theorem 3}, we can also derive the lower bound of achievable transmission rate of the sensor. For brevity, the characterization of the lower bound of transmission rate is provided in {\bf Appendix II}.  

\subsection{Optimal Choice of the Parameter in the Time-Switching Architecture}
\label{subsec:optimaltau}
 
In this subsection, we consider the choice of the parameter $\tau$ in the time-switching architecture. Let us recall that the time-switching architecture corresponds to the choice of $\varrho=\tau$ and $\eta=1-\tau$. First note that the bound of Theorem~\ref{thm:estimationphi} is a decreasing function of $\tau$. In other words, in the worst-case scenario,  the power outage probability is lowest when $\tau$ is largest.

Turning our attention to the result of Theorem~\ref{thm:transmissionoutage}, we note that the bound obtained therein is equal to $1$ for $\tau=0$, and also equal to $1$ for $\tau=1$. It follows that the upper bound to the transmission outage probability is minimized for a certain $\tau$ in $[0,1]$. Therefore, the value $\tau^*$ which minimizes the transmission outage probability is given by
\begin{equation}
\label{eq:optimaltau1}
\tau^*=\mathrm{argmin}_{\tau\in[0,1]}\left\{\left(\prod_{n\ge 0} \left(1+\alpha  \frac{\Gamma(n+1,  \pi\rho\inf(R,\gamma_m)^2)}{n!}\right)\right)^{-1/\alpha}\right\},
\end{equation}
where we emphasize that in \eqref{eq:optimaltau1}, $\gamma_m$ depends on $\tau$. Noticing that the function
\begin{equation*}
\tau\longmapsto\left(\prod_{n\ge 0} \left(1+\alpha  \frac{\Gamma(n+1,  \pi\rho\inf(R,\gamma_m)^2)}{n!}\right)\right)^{-1/\alpha}
\end{equation*}
is decreasing, it follows that $\tau^*$ can be obtained from the simpler expression
\begin{align}
\label{eq:optimaltau2}
\tau^*
& = \mathrm{argmax}_{\tau\in[0,1]}\left\{\sqrt{ \frac{\tau\left(h_0+(1-\tau)\xi\left(1-2^{m/((1-\tau)W)}\right)\right)}{P_{\mathrm C}h_0+(1-\tau)\sigma^2\left(2^{m/((1-\tau)W)}-1\right)}}\right\},
\end{align}
where in the last expression it becomes clear that $\tau^*$ depends only on $h_0,\xi,m,W,P_{\mathrm C}$ and $\sigma$. Note that, for a certain scenario when $h_0$, $\xi$, $W$, $P_{\mathrm C}$ and $\sigma$ are fixed, the optimal $\tau$ is only dependent on $m$. This analysis validates the observation in Fig. \ref{fig:Transoutage_TS_out} that optimal $\tau$ is not affected by density $\rho$. In addition, the analysis agrees with Fig. \ref{fig:trans_outage_alpha} that the rate requirement influences the optimal value of $\tau$.  



In Fig. \ref{fig:Transoutage_tau}, we present the numerical results for the upper bound of transmission outage probability versus $\tau$, under various rate requirement $m$ for both scenarios of out-of-band ($\xi=0$) and in-band transmission ($\xi=1$). The results are averaged over $10^6$ runs of simulation. It is obvious that the optimal $\tau$ that minimizes the transmission outage probability differs based on $m$ and $\xi$. This agrees with the above analysis. Then, we compare the optimal values of $\tau$ observed from Fig. \ref{fig:Transoutage_tau} with the corresponding analytical results obtained from (\ref{eq:optimaltau2}) in Table \ref{Optimal_tau}. The analytical results can provide very accurate guideline for the setting of optimal $\tau$ to perform best in a worst-case scenario. For discrete time-switching \cite{SDurrani2013}, the analytical results can be approximated to the closest practical value of $\tau$.

\begin{figure}
\centering
\includegraphics[width=0.45\textwidth]{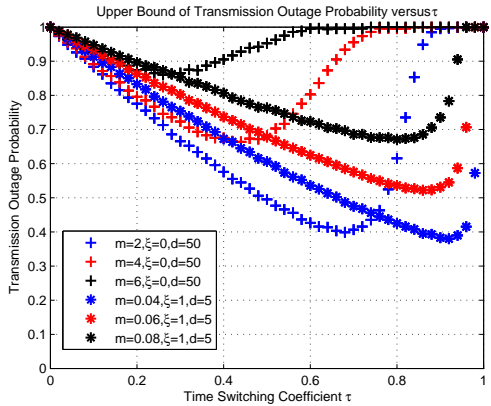}
\caption{Upper Bound of Transmission Outage Probability versus Time-Switching Coefficient $\tau$ ($\rho=0.05$, $\alpha=-1$).} \label{fig:Transoutage_tau}
\end{figure}

\begin{table}
\centering 
\caption{\footnotesize Optimal Values of $\tau$.} \label{Optimal_tau}
\begin{tabular}{cc|c|c|}
\cline{3-4}
& & Analytical result & Simulation result  \\ 
\cline{1-4}
\multicolumn{1}{ |c| }{\multirow{3}{*}{$\xi=0$} } & $ m=2 $ & 0.6828  &  0.68   \\ 
\cline{2-4}
\multicolumn{1}{ |c|  }{}    & $m=4$ &   0.4364  &   0.44 
  \\ 
\cline{2-4}
\multicolumn{1}{ |c|  }{}    &  $m=6$  &   0.2690  &   0.26     \\ 
\cline{1-4}
\multicolumn{1}{ |c|  }{\multirow{3}{*}{$\xi=1$}}   &  $m=0.02$  &   0.9185  & 0.92  \\ 
\cline{2-4}
\multicolumn{1}{ |c|  }{}   &  $m=0.06$   &    
0.8658  &  0.86        \\ 
\cline{2-4}
\multicolumn{1}{ |c|  }{}   &  $m=0.08$  &  0.7980   &     0.80   \\ 
\cline{1-4}
\end{tabular}  
\centering
\end{table}

\section{Conclusion}
\label{sec:conclusion}
This paper has presented the performance analysis of a wireless sensor powered by ambient RF energy through a stochastic geometry approach. We have analyzed the general cases when the ambient RF sources are distributed as a Ginibre $\alpha$-determinantal point process (DPP), which covers the Poisson point process case when $\alpha$ approaches zero. We  have derived the expression of the expectation and variance of the RF energy harvesting rate. We have further characterized the worst-case performance of the sensor node using the upper bound of power outage probability and transmission outage probability. Additionally, we have studied the optimal value of the time-switching coefficient to minimize the transmission outage probability. Numerical results validate all the analytical expressions, which leads us to believe that these analytical expressions are usable in practice. We have found that given a certain network density, the sensor achieves better performance when the distribution of ambient RF sources shows stronger repulsion and less attraction. Our system model can be extended by considering multiple-input and multiple-out communication channel between the sensor and data sink. Another direction of our future work is to extend the performance analysis from an individual node level to a network of nodes. 

\section*{{\bf Appendix I}}\label{app1}

\begin{proof}
Recall the following fundamental formula of point process theory found, e.g., in \cite{DaleyVereJones},
\begin{equation*}
\label{energy_variance}
V_{P_H}=\int_O\left(P^{k}_{\mathrm{H}}\right)^2\rho^{(1)}(k)\,\mathrm dk+\int_{O\times O}P^{k}_{\mathrm{H}}P^{l}_{\mathrm{H}}\,\rho^{(2)}(k,l)\,\mathrm dk\mathrm dl-\left(\int_OP^{k}_{\mathrm{H}}\,\rho^{(1)}(k)\,\mathrm dk\right)^2,
\end{equation*}
where $\rho^{(1)}$ and $\rho^{(2)}$ are the first and second correlation functions defined in \eqref{eq:correlationfunctions}. Noting that $\rho^{(2)}(k,l)=K(k,k)K(l,l)+\alpha|K(k,l)|^2$ and recombining the terms, one finds
\begin{equation*}
\mathrm{Var}\,\left(P_{\mathrm H}\right)=\int_O\left(P^{k}_{\mathrm{H}}\right)^2K(k,k)\,\mathrm dk+\alpha\int_{O\times O}P^{k}_{\mathrm{H}}P^{l}_{\mathrm{H}}\,|K(k,l)|^2\,\mathrm dk\mathrm dl.
\end{equation*}
It remains to substitute the expressions of $P^{k}_{\mathrm{H}}$ and $K(k,l)$ to obtain
\begin{multline*}
V_{P_H}=\left(\varrho \beta P_{\mathrm{S}} \frac{G_{\mathrm{S}} G_{\mathrm{H}} \lambda^{2}}{(4\pi )^{2}}\right)^2\\
\left(\int_O\left(\frac{1}{\left(\epsilon+\lVert \mathbf x_k\lVert\right)^2}\right)^2\rho\,\mathrm dk+\alpha\int_{O\times O}\frac{1}{\left(\epsilon+\lVert \mathbf x_k\lVert\right)^2}\frac{1}{\left(\epsilon+\lVert \mathbf x_l\lVert\right)^2}\,\left|\rho e^{\pi\rho k\bar{l}} e^{-\frac{\pi\rho}{2}( |k|^2 + |l|^2)}\right|^2\,\mathrm dk\mathrm dl\right).
\end{multline*}
By polar change of variable in the first term, and simply rewriting the second, we have
\begin{multline*}
 V_{P_H}=\left(\varrho \beta P_{\mathrm{S}} \frac{G_{\mathrm{S}} G_{\mathrm{H}} \lambda^{2}}{(4\pi )^{2}}\right)^2 \nonumber \\ 
\left(2\pi\rho\int_0^R\frac r{\left(\epsilon+r\right)^4}\,\mathrm dr+\alpha\rho^2\int_{O\times O}\frac{1}{\left(\epsilon+\lVert \mathbf x\lVert\right)^2}\frac{1}{\left(\epsilon+\lVert \mathbf y\lVert\right)^2}\,e^{-\pi\rho\lVert \mathbf x-\mathbf y\lVert^2}\,\mathrm d\mathbf x\mathrm d\mathbf y\right)\!.
\end{multline*}
We conclude by simply evaluating the first integral analytically and obtain \eqref{eq:variance_harvested}.
\end{proof}

\section*{{\bf Appendix II}} \label{app2}
This appendix derives the worst-case expectation of the transmission rate of the sensor, based on Theorem~\ref{thm:transmissionoutage}.

\begin{theorem}
\label{thm:evRi}
The expected transmission rate of the sensor node is lower-bounded as follows:
\begin{eqnarray}
\label{eq:evRi2}
 \mathbb{E}[C] \ge \sup_{M\in(0,+\infty)}M\left(1-\left(\prod_{n\ge 0} \left(1+\alpha  \frac{\Gamma(n+1, \pi\rho\inf(R,{\gamma}_M)^2)}{n!}\right)\right)^{-1/\alpha}\right),
\end{eqnarray}
where $\gamma_M$ was defined in \eqref{eq:gammam} which we recall here:
\begin{eqnarray}
	\gamma_M=\frac\lambda{4\pi}\sqrt{ \frac{P_{\mathrm S}\varrho\beta G_{\mathrm{S}} G_{\mathrm{H}}\left(h_0+\eta\xi\left(1-2^{M/(\eta W)}\right)\right)}{P_{\mathrm C}h_0+\eta\sigma^2\left(2^{M/(\eta W)}-1\right)}}.
\end{eqnarray}
\end{theorem}
\begin{proof}
Since $C\ge 0$, we apply Markov's inequality for any $M>0$,
\begin{equation*}
\mathbb{E}[C] \ge M\mathbb P\left(C\ge M\right),
\end{equation*}
and it suffices to use Theorem~\ref{thm:transmissionoutage} to bound the probability appearing on the r.h.s.:
\begin{eqnarray}
\label{eq:intermeqthm4}
 \mathbb{E}[C] \ge M \left(1-\left(\prod_{n\ge 0} \left(1+\alpha  \frac{\Gamma(n+1,  \pi\rho\inf(R,{\gamma}_M)^2)}{n!}\right)\right)^{-1/\alpha}\right).
\end{eqnarray}
Lastly, it suffices to remark that the r.h.s. of \eqref{eq:intermeqthm4} tends to zero as $M$ tends to zero and as $M$ tends to infinity. Since as a function of $M$, the r.h.s. of \eqref{eq:intermeqthm4} is continuous, the supremum over all $M\in(0,+\infty)$ is therefore finite.
\end{proof}


\section*{Acknowledgement}

This work was supported in part by the Singapore MOE Tier 1  (RG33/12), Singapore MOE Tier 2 (M4020140 and MOE2014-T2-2-015 ARC 4/15), and National Research Foundation of Korea (NRF) grant funded by the Korean government (MSIP) (2014R1A5A1011478).


\begin{thebibliography}{99}




\bibitem{X.2014Lu}
X. Lu, P. Wang, D. Niyato, and Z. Han, ``Resource allocation in wireless networks with RF energy harvesting and transfer," to appear in \emph{IEEE Network}.

\bibitem{XLuSurvey}
X. Lu, P. Wang, D. Niyato, D. I. Kim, and Z. Han, ``RF energy harvesting network: A contemporary survey." to appear in \emph{IEEE Communications Surveys $\&$ Tutorials}. 
 

\bibitem{Popovic2013} 
Z. Popovic, E. A. Falkenstein, D. Costinett, and R. Zane, ``Low-power far-field wireless powering for wireless sensors," \emph{Proceedings of the IEEE}, vol 101, no. 6, pp. 1397-1409, June 2013.


\bibitem{X.Lu2014}
X. Lu, P. Wang, D. Niyato, and E. Hossain, ``Dynamic Spectrum Access in Cognitive Radio Networks with RF Energy Harvesting," \emph{IEEE Wireless Communications}, vol. 21, no. 3, pp. 102-110, June 2014.


\bibitem{NParks}
A. N. Parks, A. P. Sample, Y. Zhao, J. R. Smith, ``A wireless sensing platform utilizing ambient RF energy," in \emph{Proc. of Biomedical Wireless Technologies, Networks, and Sensing Systems (BioWireleSS)}, pp. 154-156, 2013.  

\bibitem{X.Lu2015}
X. Lu, P. Wang, D. Niyato, D. I. Kim, and H. Zhu, ``Wireless Charger Networking for Mobile Devices: Fundamentals, Standards, and Applications,"  \emph{IEEE Wireless Communications}, vol. 22, no. 2, pp. 126-135, April 2015. 


\bibitem{A2009Sample}
A. Sample and J. R. Smith, ``Experimental results with two wireless power transfer systems," in \emph{Proc. of IEEE Radio Wireless Symp.}, pp. 16-18, San Diego, CA, Jan. 2009. 

\bibitem{M2008Tentzeris}
M. M. Tentzeris and Y. Kawahara, ``Novel energy harvesting technologies for ICT applications," in \emph{Proc. of IEEE Int. Symp. Appl. Internet}, pp. 373-376, Aug. 2008.


\bibitem{D2010Bouchouicha}
D. Bouchouicha, F. Dupont, M. Latrach, and L. Ventura, ``Ambient RF energy harvesting," in \emph{Proc. of International Conference on Renewable Energies and Power Quality}, pp. 1-4, March 2010.

\bibitem{R2003Shigeta}
R. Shigeta, T. Sasaki, D. M. Quan, Y. Kawahara, R. J. Vyas, M. M. Tentzeris, and T. Asami, ``Ambient RF energy harvesting sensor device with capacitor-leakage-aware duty cycle control," \emph{IEEE Sens. J.}, vol. 13, no. 8, pp. 2973-2983, July 2003.

\bibitem{V2013Liu}
V. Liu, A. Parks, V. Talla, S. Gollakota, D. Wetherall, and J. R. Smith, ``Ambient backscatter: Wireless communication out of thin air," in \emph{Proc. of ACM SIGCOMM}, Aug. 2013.

 
\bibitem{U2012Olgun} 
U. Olgun, C.-C. Chen, and J. L. Volakis, ``Design of an efficient ambient WiFi energy harvesting system," \emph{IET Microwaves, Antennas $\&$ Propagation}, vol. 6, no. 11, pp. 1200-1206, August 2012. 

 
\bibitem{M2013Pinuela} 
M. Pinuela, P. D. Mitcheson, and S. Lucyszyn, ``Ambient RF energy harvesting in urban and semi-urban environments," \emph{IEEE Transactions on Microwave Theory and Techniques}, vol. 61, no. 7, pp. 2715-2726, July 2013. 
 

\bibitem{P2013Nintanavongsa} 
P. Nintanavongsa, M. Y. Naderi, and K. R. Chowdhury, ``A dual-band wireless energy transfer protocol for heterogeneous sensor networks powered by RF energy harvesting," in \emph{Proc. of IEEE International Computer Science and Engineering Conference (ICSEC)}, pp. 387-392,
Nakorn Pathom, Thailand, Sept. 2013. 
 
 
\bibitem{G2011Karthik} 
B. G. Karthik, S. Shivaraman, and V. Aditya, ``Wi-Pie: Energy harvesting in mobile electronic devices,” in \emph{Proc. of IEEE Global Humanitarian Technology Conference (GHTC)}, pp. 398-401, Seattle, WA, Nov. 2011. 

\bibitem{K2014Huang} 
K. Huang and Vincent K.N. Lau, ``Enabling wireless power transfer in cellular networks: Architecture, modeling and deployment," \emph{IEEE Transactions on Wireless Communications}, vol. 13, no. 2, pp. 902-912, February 2014. 

\bibitem{S2013Lee}  
S. Lee, R. Zhang, and K. Huang, ``Opportunistic wireless energy harvesting in cognitive radio networks," \emph{IEEE Transactions on Wireless Communications},  vol. 12, no. 9, pp. 4788-4799,
Sept. 2013.

\bibitem{I2014Krikidis} 
I. Krikidis, ``Simultaneous information and energy transfer in large-scale networks with/without relaying," \emph{IEEE Transactions on Communications}, vol. 62, no. 3, pp. 900-912, March 2014. 

\bibitem{Z.2014Ding}
Z. Ding, I. Krikidis, B. Sharif, and H. V. Poor, ``Wireless Information and Power Transfer in Cooperative Networks With Spatially Random Relays," \emph{IEEE Transactions on Wireless Communications}, 
vol. 13, no. 8, pp. 4440-4453, Aug. 2014.

\bibitem{V.2014Mekikis}
P. V. Mekikis, A. S. Lalos, A. Antonopoulos, L. Alonso, and C. Verikoukis, ``Wireless Energy Harvesting in Two-Way Network Coded Cooperative Communications: A Stochastic Approach for Large Scale Networks," \emph{IEEE Communications Letters}, vol. 18, no. 6, pp. 1011-1014, June  2014. 



\bibitem{K.2014Huang}
K. Huang, M. Kountouris, Victor O. K. Li, ``Coverage of renewable powered cellular networks," in Proc. of \emph{IEEE International Conference on Communication Systems (ICCS)}, Macau, China, Nov. 2014  

\bibitem{Y.2014Song} 
Y. Song, M. Zhao, W. Zhou, and H. Han, ``Throughput-optimal user association in energy harvesting relay-assisted cellular networks," 
 in Proc. of \emph{IEEE International Conference on 
Wireless Communications and Signal Processing (WCSP)}, Hefei, China, Oct. 2014.


\bibitem{S.2014Dhillon}
H. S. Dhillon, Y. Li, P. Nuggehalli, Z. Pi, and J. G. Andrews, ``Fundamentals of Heterogeneous Cellular Networks with Energy Harvesting," \emph{IEEE Transactions on Wireless Communications},  
vol. 13, no. 5, pp. 2782-2797, May 2014. 

\bibitem{X.LuWCNC}
X. Lu, I. Flint, D. Niyato, N. Privault, and P. Wang, ``Performance Analysis for Simultaneously Wireless Information and Power Transfer with Ambient RF Energy Harvesting," in \emph{Proc. of IEEE WCNC}, New Orleans, LA, USA, March 2015.


\bibitem{DecreusefondFlintVergne}
L.~Decreusefond, I.~Flint, A.~Vergne, ``Efficient simulation of the Ginibre point process." (available online at {\em arXiv:1310.0800})  



\bibitem{A.A2007Abbasi} 
T. Lin, H. Santoso, and K. Wu, ``Global Sensor Deployment and Local Coverage-aware Recovery Schemes for Smart Environments," to appear in 
\emph{IEEE Transactions on Mobile Computing}.  


 
\bibitem{S2013Cho}
S. Cho, and W. Choi, ``Energy-efficient repulsive cell activation for heterogeneous cellular networks," \emph{IEEE Journal on Selected Areas in Communications}, vol. 31, no. 5, pp. 870-882, May 2013. 
 
\bibitem{N2014Vastardis}
N. Vastardis, and K. Yang, ``An enhanced community-based mobility model for distributed mobile social networks," \emph{Journal of Ambient Intelligence and Humanized Computing}, vol. 5, no. 1, pp. 65-75, February 2014. 

\bibitem{N.2012Miyoshi}
N. Miyoshi and T. Shirai, ``A cellular network model with Ginibre configurated base stations," \emph{Research Reports on Mathematical and Computing Sciences}, 2012.

\bibitem{N.Deng2014}
N. Deng, W. Zhou, and M. Haenggi, ``The Ginibre point process as a model for wireless networks with repulsion," to appear in {\em Transactions on Wireless Communications}.


\bibitem{ZhangRuiMIMO}
R.~Zhang and C.~K.~Ho, ``MIMO broadcasting for simultaneous wireless information and power transfer," \emph{IEEE Transactions on Wireless Communications}, vol. 12 , no. 5, pp. 1989-2001, May 2013. 

\bibitem{ConfVersion}
I.~Flint, X.~Lu, N.~Privault, D.~Niyato, P.~Wang, ``Performance Analysis of Ambient RF Energy Harvesting: A Stochastic Geometry Approach," to appear in \emph{IEEE GLOBECOM 2014}.


\bibitem{Visser2013}
H. J. Visser and R. J. M. Vullers, ``RF energy harvesting and transport for wireless sensor network applications: principles and requirements," \emph{Proceedings of the IEEE}, vol. 101, no. 6, pp. 1410-1423, June 2013.




 






\bibitem{Kallenberg}
O.~Kallenberg, { \em Random measures,} fourth ed. Berlin, Germany: Akademie-Verlag, 1986.

\bibitem{G2009Miao}
G. Miao, N. Himayat, Y. G. Li, and A. Swami, ``Cross-layer optimization for energy-efficient wireless communications: A survey," \emph{Wireless Commun. Mobile Comput.}, vol. 9, pp. 529-542, Apr. 2009.


\bibitem{DaleyVereJones}
D.~J. Daley and D.~Vere-Jones, {\em An introduction to the theory of point processes. Vol. I,} 2nd ed. New York: Springer-Verlag, Probability and its Applications,  2003. 

\bibitem{Brezis}
H.~Brezis, {\em Analyse fonctionnelle,} Paris, France: Masson, Collection Math\'ematiques Appliqu\'ees pour la Ma\^\i trise.
[Collection of Applied Mathematics for the Master's Degree], 1983. 

\bibitem{GeorgiiYoo}
H.~Georgii and H.~J. Yoo, ``Conditional intensity and Gibbsianness of determinantal point processes," {\em J. Stat. Phys.}, vol. 118, vol. 1, pp. 55-84, 2005.

\bibitem{ShiraiTakahashi}
T.~Shirai and Y.~Takahashi, ``Random point fields associated with certain Fredholm determinants I: Fermion, Poisson and boson point processes," in {\em Ann. Probab.}, vol. 31, no. 3, pp. 1533-1564, Dec. 2003.

\bibitem{Soshnikov}
A.~Soshnikov. \newblock ``Determinantal random point fields." \newblock {\em Uspekhi Mat. Nauk}, vol. 55, no. 5, pp. 107-160, 2000.


\bibitem{Hough}
J.~B. Hough, M.~Krishnapur, Y.~Peres, and B.~Vir{\'a}g, ``Determinantal processes and independence," {\em Probab. Surv.}, vol. 3, pp. 206-229, (electronic), 2006.

\bibitem{DecreusefondFlintPrivaultTorrisi}
L.~Decreusefond, I.~Flint, N.~Privault, G.L.~Torrisi, ``Determinantal processes: A survey," in {\em Stochastic analysis for Poisson point processes: Malliavin calculus, Wiener-It\^o chaos expansions and stochastic geometry}, G.~Peccati and M.~Reitzner editors, Bocconi $\&$ Springer Series, 2014.

\bibitem{DecreusefondFlintPrivaultTorrisi2}
L.~Decreusefond, I.~Flint, N.~Privault, G.L.~Torrisi, ``Stochastic dynamics of determinantal processes by integration by parts."  (available onlin at {\em arXiv:1210.6109}).

\bibitem{N.2013Parks}
A. N. Parks, A. P. Sample, Y. Zhao, J. R. Smith, ``A Wireless Sensing
Platform Utilizing Ambient RF Energy," \emph{in Proc. of Biomedical Wireless Technologies, Networks, and Sensing Systems (BioWireleSS)}, Austin, TX, Jan. 2013.

\bibitem{2014X.Lu}
X. Lu, P. Wang, and D. Niyato, ``A Layered Coalitional Game Framework of Wireless Relay Network," \emph{IEEE Transactions on Vehicular Technology}, vol. 63, no. 1, pp. 472-478, January 2014. 


\bibitem{SDurrani2013}
A.~A.~Nasir, X.~Zhou, S.~Durrani, and R.~A.~Kennedy, ``Relaying protocols for wireless energy harvesting and information processing," \emph{IEEE Transactions on Wireless Communications}, vol. 12, no. 7, pp. 3622-3636, July 2013.


 






 













 


 


\end{thebibliography}
\end{document}